\documentclass[12pt,reqno]{amsart}
\usepackage{graphicx}
\usepackage{amscd}
\usepackage{amsmath}
\usepackage{epsfig}
\usepackage{amsfonts}
\usepackage{amssymb}

\providecommand{\U}[1]{\protect\rule{.1in}{.1in}}
\providecommand{\U}[1]{\protect\rule{.1in}{.1in}}
\allowdisplaybreaks
\newtheorem{theorem}{Theorem}
\theoremstyle{plain}

\newtheorem{definition}{Definition}

\newtheorem{lemma}{Lemma}

\numberwithin{equation}{section}

\input{tcilatex}

\begin{document}
\title[Quantum Integrals of Motion]{Quantum Integrals of Motion for Variable
Quadratic Hamiltonians}
\author{Ricardo Cordero-Soto}
\address{Mathematical, Computational and Modeling Sciences Center, Arizona
State University, Tempe, AZ 85287--1804, U.S.A.}
\email{ricardojavier81@gmail.com}
\author{Erwin Suazo}
\address{Department of Mathematical Sciences, University of Puerto Rico,
Mayaquez, call box 9000, PR 00681--9000, Puerto Rico}
\email{erwin.suazo@upr.edu}
\author{Sergei K. Suslov}
\address{School of Mathematical and Statistical Sciences \& Mathematical,
Computational and Modeling Sciences Center, Arizona State University, Tempe,
AZ 85287--1804, U.S.A.}
\email{sks@asu.edu}
\urladdr{http://hahn.la.asu.edu/\symbol{126}suslov/index.html}
\date{\today }
\subjclass{Primary 81Q05, 35C05. Secondary 42A38}
\keywords{The time-dependent Schr\"{o}dinger equation, Cauchy initial value
problem, Green function, propagator, quantum damped oscillators,
Caldirola-Kanai Hamiltonians, quantum integrals of motion, Lewis--Riesenfeld
dynamical invariant, Ermakov's equation, Ehrenfest's theorem}

\begin{abstract}
We construct integrals of motion for several models of the quantum damped
oscillators in a framework of a general approach to the time-dependent Schr%
\"{o}dinger equation with variable quadratic Hamiltonians. An extension of
the Lewis--Riesenfeld dynamical invariant is given. The time-evolution of
the expectation values of the energy related positive operators is
determined for the oscillators under consideration. A proof of uniqueness of
the corresponding Cauchy initial value problem is discussed as an
application.
\end{abstract}

\maketitle

\section{An Introduction}

Evolution of a nonrelativistic quantum system from a given initial state to
the final state is governed by the (time-dependent) Schr\"{o}dinger
equation. Unfortunately, its explicit solutions are available only for the
simplest Hamiltonians and, in general, one has to rely on a variety of
approximation, asymptotic and numerical methods. Luckily among the
integrable cases are the so-called quadratic Hamiltonians that attracted
substantial attention over the years in view of their great importance to
many advanced quantum problems. Examples can be found in quantum and
physical optics \cite{Delgadoetal98}, \cite{Klauder:Sudarshan}, \cite%
{PadillaMaster}, \cite{Reithmaieretal97}, physics of lasers and masers \cite%
{Sargent:Scully:Lamb74}, \cite{Tarasov83}, \cite{Scully:Zubairy97}, \cite%
{Walls94}, molecular spectroscopy \cite{Dokt:Mal:Man77}, quantum chemistry,
quantization of mechanical systems \cite{Degas:Ruijsenaars01}, \cite%
{Faddeyev69}, \cite{FeynmanPhD}, \cite{Feynman}, \cite{Fey:Hib}, \cite%
{Kochan07}, \cite{Kochan10} and Hamiltonian cosmology \cite%
{Bertoni:Finelli:Venturi98}, \cite{Finelli:Gruppuso:Venturi99}, \cite%
{Finelli:Vacca:Venturi98}, \cite{Hawkins:Lidsey02}, \cite{IacobF}, \cite%
{PadillaMaster}, \cite{Rosu:Espinoza99}, \cite{Rosu:Espinoza:Reyes99}, \cite%
{Ryan72}. They include coherent states \cite{Malkin:Man'ko79}, \cite%
{Malk:Man:Trif69}, \cite{Malk:Man:Trif70}, \cite{Klauder:Sudarshan} and
Berry's phase \cite{Berry85}, \cite{Berry:Hannay88}, \cite%
{Cervero:Lejarreta89}, \cite{Hannay85}, \cite{Leach90}, \cite{Morales88},
asymptotic and numerical methods \cite{Goyaletal93}, \cite{Kamenshchiketal06}%
, \cite{Kruskal62}, \cite{Milne30}, \cite{Mun:Ru-Paz:Wolf09}, charged
particle traps \cite{Major:Gheorghe:Werth} and motion in uniform magnetic
fields \cite{Cor-Sot:Lop:Sua:Sus}, \cite{Corant:Snyder58}, \cite%
{Dodonov:Man'koFIAN87}, \cite{La:Lif}, \cite{Lewis67}, \cite{Lewis68}, \cite%
{Lewis:Riesen69}, \cite{Malk:Man:Trif70}, polyatomic molecules in varying
external fields, crystals through which an electron is passing and exciting
the oscillator modes and other interactions of the modes with external
fields \cite{Fey:Hib}. Quadratic Hamiltonians have particular applications
in quantum electrodynamics because the electromagnetic field can be
represented as a set of forced harmonic oscillators \cite{Bo:Shi}, \cite%
{Fey:Hib}, \cite{Dodonov:Man'koFIAN87}, \cite{Gottf:T-MY}, \cite%
{Ivan:Mal:Man74}, and \cite{Merz}. Nonlinear oscillators play a central role
in the novel theory of Bose--Einstein condensation \cite%
{Dal:Giorg:Pitaevski:Str99} based on the nonlinear Schr\"{o}dinger (or
Gross--Pitaevskii) equation \cite{Kagan:Surkov:Shlyap96}, \cite%
{Kagan:Surkov:Shlyap97}, \cite{Kivsh:Alex:Tur01}, \cite{Per-G:Tor:Mont}%
.\medskip\ 

The one-dimensional Schr\"{o}dinger equation with variable quadratic
Hamiltonians of the form%
\begin{equation}
i\frac{\partial \psi }{\partial t}=-a\left( t\right) \frac{\partial ^{2}\psi 
}{\partial x^{2}}+b\left( t\right) x^{2}\psi -i\left( c\left( t\right) x%
\frac{\partial \psi }{\partial x}+d\left( t\right) \psi \right) ,
\label{in1}
\end{equation}%
where $a\left( t\right) ,$ $b\left( t\right) ,$ $c\left( t\right) ,$ and $%
d\left( t\right) $ are real-valued functions of time $t$ only, can be
integrated in the following manner (see, for example, \cite%
{Cor-Sot:Lop:Sua:Sus}, \cite{Cor-Sot:Sua:Sus}, \cite{Cor-Sot:Sus}, \cite%
{Dod:Mal:Man75}, \cite{Lan:Sus}, \cite{Lop:Sus}, \cite{Me:Co:Su}, \cite%
{SuazoF}, \cite{Sua:Sus}, \cite{Suaz:Sus}, \cite{Sua:Sus:Vega}, \cite{Wolf81}%
, and \cite{Yeon:Lee:Um:George:Pandey93} for a general approach and some
elementary solutions). The Green functions, or Feynman's propagators, are
given by \cite{Cor-Sot:Lop:Sua:Sus}, \cite{Suaz:Sus}:%
\begin{equation}
\psi =G\left( x,y,t\right) =\frac{1}{\sqrt{2\pi i\mu \left( t\right) }}\
e^{i\left( \alpha \left( t\right) x^{2}+\beta \left( t\right) xy+\gamma
\left( t\right) y^{2}\right) },  \label{in2}
\end{equation}%
where%
\begin{eqnarray}
&&\alpha \left( t\right) =\frac{1}{4a\left( t\right) }\frac{\mu ^{\prime
}\left( t\right) }{\mu \left( t\right) }-\frac{d\left( t\right) }{2a\left(
t\right) },  \label{in3} \\
&&\beta \left( t\right) =-\frac{h\left( t\right) }{\mu \left( t\right) }%
,\qquad h\left( t\right) =\exp \left( -\int_{0}^{t}\left( c\left( \tau
\right) -2d\left( \tau \right) \right) \ d\tau \right) ,  \label{in4} \\
&&\gamma \left( t\right) =\frac{a\left( t\right) h^{2}\left( t\right) }{\mu
\left( t\right) \mu ^{\prime }\left( t\right) }+\frac{d\left( 0\right) }{%
2a\left( 0\right) }-4\int_{0}^{t}\frac{a\left( \tau \right) \sigma \left(
\tau \right) h^{2}\left( \tau \right) }{\left( \mu ^{\prime }\left( \tau
\right) \right) ^{2}}\ d\tau  \label{in5}
\end{eqnarray}%
and the function $\mu \left( t\right) $ satisfies the so-called\
characteristic equation%
\begin{equation}
\mu ^{\prime \prime }-\tau \left( t\right) \mu ^{\prime }+4\sigma \left(
t\right) \mu =0  \label{in6}
\end{equation}%
with%
\begin{equation}
\tau \left( t\right) =\frac{a^{\prime }}{a}-2c+4d,\qquad \sigma \left(
t\right) =ab-cd+d^{2}+\frac{d}{2}\left( \frac{a^{\prime }}{a}-\frac{%
d^{\prime }}{d}\right)  \label{in7}
\end{equation}%
subject to the initial data%
\begin{equation}
\mu \left( 0\right) =0,\qquad \mu ^{\prime }\left( 0\right) =2a\left(
0\right) \neq 0.  \label{in8}
\end{equation}%
(More details can be found in Refs.~\cite{Cor-Sot:Lop:Sua:Sus}, \cite%
{Suaz:Sus} and a Hamiltonian structure is considered in Refs.~\cite{Berry85}%
, \cite{Cor-Sot:Sus}.) Then, by the superposition principle, solution of the
Cauchy initial value problem can be presented in an integral form%
\begin{equation}
\psi \left( x,t\right) =\int_{-\infty }^{\infty }G\left( x,y,t\right) \
\varphi \left( y\right) \ dy,\quad \lim_{t\rightarrow 0^{+}}\psi \left(
x,t\right) =\varphi \left( x\right)  \label{CauchyInVProb}
\end{equation}%
for a suitable initial function $\varphi $ on $%
\mathbb{R}
$ (a rigorous proof is given in Ref.~\cite{Suaz:Sus} and uniqueness is
analyzed in this paper).\medskip

We discuss integrals of motion for several particular models of the damped
and generalized quantum oscillators. The simple harmonic oscillator is of
interest in many quantum problems \cite{Fey:Hib}, \cite{La:Lif}, \cite{Merz}%
, and \cite{Schiff}. The forced harmonic oscillator was originally
considered by Richard Feynman in his path integrals approach to the
nonrelativistic quantum mechanics \cite{FeynmanPhD}, \cite{Feynman}, \cite%
{Feynman49a}, \cite{Feynman49b}, and \cite{Fey:Hib}; see also \cite{Lop:Sus}%
. Its special and limiting cases were discussed in Refs.~\cite{Beauregard}, 
\cite{Gottf:T-MY}, \cite{Holstein}, \cite{Maslov:Fedoriuk}, \cite{Merz}, 
\cite{Thomber:Taylor} for the simple harmonic oscillator and in Refs.~\cite%
{Arrighini:Durante}, \cite{Brown:Zhang}, \cite{Holstein97}, \cite{Nardone}, 
\cite{Robinett} for the particle in a constant external field; see also
references therein. The damped oscillations have been studied to a great
extent in classical mechanics~\cite{Bateman31}, \cite{BatemanPDE} and \cite%
{Lan:Lif}. Their quantum analogs are introduced and analyzed from different
viewpoints by many authors; see, for example, \cite{Caldirola41}, \cite%
{Chand:Senth:Laksh07}, \cite{ChruAnnPhys06}, \cite{Chru:JurkAnnPhys06}, \cite%
{Chru:Jurk}, \cite{Cor-Sot:Sua:Sus}, \cite{Denman66}, \cite{Dekker81}, \cite%
{Dito:Turr06}, \cite{Dod:Man79}, \cite{Dod:Miz:Dod}, \cite{LeachAmJPhys78}, 
\cite{LeachSIAM78}, \cite{Kanai48}, \cite{Mont03}, \cite{Nieto:Truax}, \cite%
{Svin75}, \cite{Svin76}, \cite{Tarasov01}, \cite{Um:Yeon:George}, and
references therein. The quantum parametric oscillator with variable
frequency is also largely studied in view of its physical importance; see,
for example, \cite{Cher:Man08}, \cite{Dodonov:Man'koFIAN87}, \cite{HusimiI53}%
, \cite{HusimiII53}, \cite{Lan:Sus}, \cite{Malkin:Man'ko79}, \cite%
{Malk:Man:Trif70}, \cite{Perelomov:Popov69}, \cite{Per:Zel}, \cite%
{Popov:PerelomovI69}, \cite{Popov:PerelomovII69}, \cite{Schuch08}, and \cite%
{Solimenoetal69}; a detailed bibliography is given in \cite{Camizetall71}%
.\medskip

In the present paper we revisit a familiar topic of the quantum integrals of
motion for the time-dependent Schr\"{o}dinger equation%
\begin{equation}
i\frac{\partial \psi }{\partial t}=H\left( t\right) \psi  \label{in9}
\end{equation}%
with variable quadratic Hamiltonians of the form 
\begin{equation}
H=a\left( t\right) p^{2}+b\left( t\right) x^{2}+d\left( t\right) \left(
px+xp\right) ,  \label{in10}
\end{equation}%
where $p=-i\partial /\partial x,$ $\hslash =1$ and $a\left( t\right) ,$ $%
b\left( t\right) ,$ $c\left( t\right) =2d\left( t\right) $ are some
real-valued functions of time only (see, for example, \cite{Dod:Mal:Man75}, 
\cite{Leach90}, \cite{Lewis:Riesen69}, \cite{Malk:Man:Trif70}, \cite%
{Malk:Man:Trif73}, \cite{Wolf81}, \cite{Yeon:Lee:Um:George:Pandey93} and
references therein). A related energy operator $E$ is defined in a
traditional way as a quadratic in $p$ and $x$ operator that has constant
expectation values \cite{Dodonov:Man'koFIAN87}:%
\begin{equation}
\frac{d}{dt}\left\langle E\right\rangle =\frac{d}{dt}\int_{-\infty }^{\infty
}\psi ^{\ast }E\psi \ dx=0.  \label{in11}
\end{equation}%
It is well-known that such quadratic invariants are not unique. Although an
elegant general solution is known, say, for the parametric oscillator, it
involves an integration of nonlinear Ermakov's equation \cite{Lewis:Riesen69}%
. Here the simplest energy operators are constructed for several integrable
models of the damped and modified quantum oscillators. Then an extension of
the familiar Lewis--Riesenfeld quadratic invariant is given to the most
general case of the variable non-self-adjoint quadratic Hamiltonian (see
also \cite{Leach90}, \cite{Wolf81}, \cite{Yeon:Lee:Um:George:Pandey93}, we
do not use canonical transformations and deal only with real-valued
solutions of the corresponding generalized Ermakov system), which seems to
be missing in the available literature and may be considered as the main
result of this paper. (An attempt to collect relevant references is made%
\footnote{%
A complete bibliography on classical and quantum generalized harmonic
oscillators, their invariants, group-theoretical methods and applications is
very extensive. Only case of the damped oscillators in \cite{Dekker81}
includes about 600 references!}.) Group-theoretical aspects will be
discussed elsewhere, we only provide the factorization of the general
quadratic invariant (see also \cite{Suslov10}).\medskip\ 

In general the average $\left\langle E\right\rangle $ is not positive. A
complete dynamics of the expectation values of some energy-related positive
operators is found instead for each model, which is a somewhat interesting
result on its own. In addition to other works \cite{Berry85}, \cite%
{Dodonov:Man'koFIAN87}, \cite{Dod:Mal:Man75}, \cite{Hannay85}, \cite%
{Lewis:Riesen69}, \cite{Malkin:Man'ko79}, \cite{Malk:Man:Trif73}, \cite%
{Wolf81}, \cite{Yeon:Lee:Um:George:Pandey93} these advances allow us to
discuss uniqueness of the corresponding Cauchy initial value problem for the
special models and for the general quadratic Hamiltonian under consideration
as a modest contribution to this well-developed area of quantum mechanics
and partial differential equations. Further relations of the quadratic
invariants with the solution of the initial value problem are discussed in
the forthcoming paper \cite{Suslov10}.\medskip

The paper is organized as follows. In Section~2 we review several exactly
solvable models of the damped and generalized oscillators in quantum
mechanics. Some of these \textquotedblleft exotic\textquotedblright\
oscillators with variable quadratic Hamiltonians appear to be missing,
and/or are just recently introduced, in the available literature. The
corresponding Green functions are found in terms of elementary functions.
The dynamical invariants and quadratic energy-related operators are
discussed in Sections~3 and 4. The last section is concerned with an
application to the Cauchy initial value problems. The classical equations of
motion for the expectation values of the position operator for the quantum
oscillators under consideration are derived in Appendix~A. The Heisenberg
uncertainty relation and linear dynamic invariants are revisited,
respectively, in Appendices~B and C. Solutions of a required differential
equation are given in Appendix~D to make our presentation is as
self-contained as possible.

\section{Some Integrable Quadratic Hamiltonians}

Quantum systems with the Hamiltonians (\ref{in10}) are called the
generalized harmonic oscillators \cite{Berry85}, \cite{Dod:Mal:Man75}, \cite%
{Hannay85}, \cite{Leach90}, \cite{Wolf81}, \cite{Yeon:Lee:Um:George:Pandey93}%
. In this paper we concentrate, among others, on the following variable
Hamiltonians: the Caldirola-Kanai Hamiltonian of the quantum damped
oscillator \cite{Caldirola41}, \cite{Dekker81}, \cite{Kanai48}, \cite%
{Um:Yeon:George} and some of its natural modifications, a modified
oscillator introduced by Meiler, Cordero-Soto and Suslov \cite{Me:Co:Su}, 
\cite{Cor-Sot:Sus}, the quantum damped oscillator of Chru\'{s}ci\'{n}ski and
Jurkowski \cite{Chru:Jurk} in the coordinate and momentum representations
and a quantum-modified parametric oscillator which is believed to be new.
The Green functions are derived in a united way.

\subsection{The Caldirola-Kanai Hamiltonian}

A model of the quantum damped oscillator with a variable Hamiltonian of the
form%
\begin{equation}
H=\frac{\omega _{0}}{2}\left( e^{-2\lambda t}\ p^{2}+e^{2\lambda t}\
x^{2}\right)  \label{CKham}
\end{equation}%
is called the Caldirola-Kanai model \cite{Bateman31}, \cite{Caldirola41}, 
\cite{Dekker81}, \cite{Kanai48}, \cite{Um:Yeon:George}. Nowadays it is a
standard way of adding friction to the quantum harmonic oscillator. The
Green function is given by%
\begin{equation}
G\left( x,y,t\right) =\sqrt{\frac{\omega e^{\lambda t}}{2\pi i\omega
_{0}\sin \omega t}}\ e^{i\left( \alpha \left( t\right) x^{2}+\beta \left(
t\right) xy+\gamma \left( t\right) y^{2}\right) },\quad \omega =\sqrt{\omega
_{0}^{2}-\lambda ^{2}}>0,  \label{sm1}
\end{equation}%
where%
\begin{eqnarray}
\alpha \left( t\right) &=&\frac{\omega \cos \omega t-\lambda \sin \omega t}{%
2\omega _{0}\sin \omega t}e^{2\lambda t},  \label{sm2} \\
\beta \left( t\right) &=&-\frac{\omega }{\omega _{0}\sin \omega t}e^{\lambda
t},  \label{sm3} \\
\gamma \left( t\right) &=&\frac{\omega \cos \omega t+\lambda \sin \omega t}{%
2\omega _{0}\sin \omega t}.  \label{sm4}
\end{eqnarray}%
This popular model had been studied in detail by many authors from different
viewpoints; see, for example, \cite{Antonsen}, \cite{Britt50}, \cite%
{Cari:Luc:Ra08}, \cite{Cari:Luc:Ra09}, \cite{Caval98}, \cite{Cheng84}, \cite%
{Cheng85}, \cite{Dod:Man79}, \cite{Karavayev}, \cite{Kh:Am06}, \cite%
{Kim:Sant:Khan02}, \cite{Kochan07}, \cite{Kochan10}, \cite{LeachAmJPhys78}, 
\cite{Nieto:Truax}, \cite{Oh:Lee:George89}, \cite{Ped:Gue03}, \cite{Safonov}%
, \cite{Svin75}, \cite{Svin76}, \cite{Tikoch78}, \cite{Yeon:Um:George} and
references therein, a detailed bibliography can be found in \cite{Dekker81}, 
\cite{Um:Yeon:George}.

\subsection{A Modified Caldirola-Kanai Hamiltonian}

In this paper, we would like to consider another version of the quantum
damped oscillator with variable Hamiltonian of the form%
\begin{equation}
H=\frac{\omega _{0}}{2}\left( e^{-2\lambda t}\ p^{2}+e^{2\lambda t}\
x^{2}\right) -\lambda \left( px+xp\right) .  \label{modCKham}
\end{equation}%
The Green functions in (\ref{sm1}) has%
\begin{eqnarray}
\alpha \left( t\right) &=&\frac{\omega \cos \omega t+\lambda \sin \omega t}{%
2\omega _{0}\sin \omega t}e^{2\lambda t},  \label{sm5} \\
\beta \left( t\right) &=&-\frac{\omega }{\omega _{0}\sin \omega t}e^{\lambda
t},  \label{sm6} \\
\gamma \left( t\right) &=&\frac{\omega \cos \omega t-\lambda \sin \omega t}{%
2\omega _{0}\sin \omega t}.  \label{sm7}
\end{eqnarray}%
This can be derived directly from equations (\ref{in2})--(\ref{in8})
following Refs.~\cite{Cor-Sot:Lop:Sua:Sus} and \cite{Cor-Sot:Sua:Sus}%
.\medskip

The Ehrenfest theorem for both Caldirola-Kanai models has the same form%
\begin{equation}
\frac{d^{2}}{dt^{2}}\left\langle x\right\rangle +2\lambda \frac{d}{dt}%
\left\langle x\right\rangle +\omega _{0}^{2}\left\langle x\right\rangle =0,
\label{cldamposc}
\end{equation}%
which coincides with the classical equation of motion for a damped
oscillator \cite{BatemanPDE}, \cite{Lan:Lif}. Details are provided in
Appendix~A.

\subsection{The United Model}

The following non-self-adjoint Hamiltonian:%
\begin{equation}
H=\frac{\omega _{0}}{2}\left( e^{-2\lambda t}\ p^{2}+e^{2\lambda t}\
x^{2}\right) -\mu xp  \label{UMHam}
\end{equation}%
coincides with the original Caldirola-Kanai model when $\mu =0$ and the
Hamiltonian is self-adjoint. Another special case $\lambda =0$ corresponds
to the quantum damped oscillator discussed in \cite{Cor-Sot:Sua:Sus} as an
example of a simple quantum system with the non-self-adjoint Hamiltonian.
(This is an alternative way to introduce dissipation of energy to the
quantum harmonic oscillator.) Combining both cases we refer to (\ref{UMHam})
as the united Hamiltonian.\medskip

The Green function is given by%
\begin{equation}
G\left( x,y,t\right) =\sqrt{\frac{\omega e^{\left( \lambda -\mu \right) t}}{%
2\pi i\omega _{0}\sin \omega t}}\ e^{i\left( \alpha \left( t\right)
x^{2}+\beta \left( t\right) xy+\gamma \left( t\right) y^{2}\right) },
\label{UMGreen}
\end{equation}%
where%
\begin{eqnarray}
\alpha \left( t\right) &=&\frac{\omega \cos \omega t+\left( \mu -\lambda
\right) \sin \omega t}{2\omega _{0}\sin \omega t}e^{2\lambda t},
\label{UMGreenA} \\
\beta \left( t\right) &=&-\frac{\omega }{\omega _{0}\sin \omega t}e^{\lambda
t},  \label{UMGreenB} \\
\gamma \left( t\right) &=&\frac{\omega \cos \omega t+\left( \lambda -\mu
\right) \sin \omega t}{2\omega _{0}\sin \omega t}  \label{UMGreenC}
\end{eqnarray}%
with $\omega =\sqrt{\omega _{0}^{2}-\left( \lambda -\mu \right) ^{2}}>0.$%
\medskip

In this case the Ehrenfest theorem takes the form:%
\begin{equation}
\frac{d^{2}}{dt^{2}}\left\langle x\right\rangle +\ 2\left( \lambda +\mu
\right) \frac{d}{dt}\left\langle x\right\rangle +\left( \omega
_{0}^{2}+4\lambda \mu \right) \left\langle x\right\rangle =0.
\label{UMEhren}
\end{equation}%
It is derived in Appendix~A and the Heisenberg uncertainty relation is
discussed in Appendix~B.

\subsection{A Modified Oscillator}

The one-dimensional Hamiltonian of a modified oscillator introduced by
Meiler, Cordero-Soto and Suslov \cite{Me:Co:Su}, \cite{Cor-Sot:Sus} has the
form%
\begin{eqnarray}
H &=&\left( \cos t\ p+\sin t\ x\right) ^{2}  \label{mod1} \\
&=&\cos ^{2}t\ p^{2}+\sin ^{2}t\ x^{2}+\sin t\cos t\ \left( px+xp\right) 
\notag \\
&=&\frac{1}{2}\left( p^{2}+x^{2}\right) +\frac{1}{2}\cos 2t\ \left(
p^{2}-x^{2}\right) +\frac{1}{2}\sin 2t\ \left( px+xp\right) .  \notag
\end{eqnarray}%
(A physical interpretation of this Hamiltonian from the viewpoint of quantum
dynamical invariants will be discussed in Section~4.) The Green function is
given in terms of trigonometric and hyperbolic functions as follows%
\begin{eqnarray}
G\left( x,y,t\right) &=&\frac{1}{\sqrt{2\pi i\left( \cos t\sinh t+\sin
t\cosh t\right) }}  \label{modGreen} \\
&&\times \exp \left( \frac{\left( x^{2}-y^{2}\right) \sin t\sinh
t+2xy-\left( x^{2}+y^{2}\right) \cos t\cosh t}{2i\left( \cos t\sinh t+\sin
t\cosh t\right) }\right) .  \notag
\end{eqnarray}%
More details can be found in\ \cite{Me:Co:Su}, \cite{Cor-Sot:Sus}. The
corresponding Ehrenfest theorem, namely,%
\begin{equation}
\frac{d^{2}}{dt^{2}}\left\langle x\right\rangle +\ 2\tan t\frac{d}{dt}%
\left\langle x\right\rangle -2\left\langle x\right\rangle =0,
\label{modEhrenfest}
\end{equation}%
is derived in Appendix~A.

\subsection{The Modified Damped Oscillator}

The time-dependent Schr\"{o}dinger equation%
\begin{equation}
i\hslash \frac{\partial \psi }{\partial t}=H\left( t\right) \psi
\label{SCHEQ}
\end{equation}%
with the variable quadratic Hamiltonian of the form%
\begin{equation}
H=\frac{p^{2}}{2m\cosh ^{2}\left( \lambda t\right) }+\frac{m\omega _{0}^{2}}{%
2}\cosh ^{2}\left( \lambda t\right) \ x^{2},\quad p=\frac{\hslash }{i}\frac{%
\partial }{\partial x}  \label{CJHam}
\end{equation}%
has been recently considered by Chru\'{s}ci\'{n}ski and Jurkowski \cite%
{Chru:Jurk} as a model of the quantum damped oscillator; see also \cite%
{Most98}.\medskip

In this case the characteristic equation (\ref{in6}) takes the form%
\begin{equation}
\mu ^{\prime \prime }+2\lambda \tanh \left( \lambda t\right) \mu ^{\prime
}+\omega _{0}^{2}\mu =0.  \label{CJcheq}
\end{equation}%
The particular solution is given by%
\begin{equation}
\mu \left( t\right) =\frac{\hslash }{m\omega }\frac{\sin \left( \omega
t\right) }{\cosh \left( \lambda t\right) },\qquad \omega =\sqrt{\omega
_{0}^{2}-\lambda ^{2}}>0  \label{CJmu}
\end{equation}%
and the corresponding propagator can be presented as follows%
\begin{equation}
G\left( x,y,t\right) =\sqrt{\frac{m\omega \cosh \left( \lambda t\right) }{%
2\pi i\hslash \sin \left( \omega t\right) }}\ e^{i\left( \alpha \left(
t\right) x^{2}+\beta \left( t\right) xy+\gamma \left( t\right) y^{2}\right)
},  \label{CJGreen}
\end{equation}%
where%
\begin{equation}
\alpha \left( t\right) =\frac{m\cosh \left( \lambda t\right) }{2\hslash \sin
\left( \omega t\right) }\left( \omega \cos \left( \omega t\right) \cosh
\left( \lambda t\right) -\lambda \sin \left( \omega t\right) \sinh \left(
\lambda t\right) \right) ,  \label{CJa}
\end{equation}%
\begin{equation}
\beta \left( t\right) =-\frac{m\omega \cosh \left( \lambda t\right) }{%
2\hslash \sin \left( \omega t\right) },  \label{CJb}
\end{equation}%
\begin{equation}
\gamma \left( t\right) =\frac{m\omega \cos \left( \omega t\right) }{2\hslash
\sin \left( \omega t\right) }.  \label{CJc}
\end{equation}%
(We somewhat simplify the original propagator found in \cite{Chru:Jurk}; see
also \cite{Kochan}.) This Green function can be independently derived from
our equations (\ref{in3})--(\ref{in5}) with the help of the following
elementary antiderivative:%
\begin{eqnarray}
&&\left( \frac{\lambda \cos \left( \omega t+\delta \right) \sinh \left(
\lambda t\right) +\omega \sin \left( \omega t+\delta \right) \cosh \left(
\lambda t\right) }{\omega \cos \left( \omega t+\delta \right) \cosh \left(
\lambda t\right) -\lambda \sin \left( \omega t+\delta \right) \sinh \left(
\lambda t\right) }\right) ^{\prime }  \label{CJanti} \\
&&\quad =\frac{\omega \omega _{0}^{2}\cosh ^{2}\left( \lambda t\right) }{%
\left( \omega \cos \left( \omega t+\delta \right) \cosh \left( \lambda
t\right) -\lambda \sin \left( \omega t+\delta \right) \sinh \left( \lambda
t\right) \right) ^{2}}.  \notag
\end{eqnarray}%
Further details are left to the reader.\medskip

Special cases are as follows: when $\lambda =0,$ one recovers the standard
propagator for the linear harmonic oscillator \cite{Fey:Hib}, and $\omega
_{0}=0$ gives a pure damping case \cite{Kochan}:%
\begin{equation}
G\left( x,y,t\right) =\sqrt{\frac{m\lambda }{2\pi i\hslash \tanh \left(
\omega t\right) }}\exp \left( \frac{im\lambda \left( x-y\right) ^{2}}{%
2\hslash \tanh \left( \omega t\right) }\right) .  \label{CJpure}
\end{equation}%
In the limit $\lambda \rightarrow 0$ formula (\ref{CJpure}) reproduces the
propagator for a free particle \cite{Fey:Hib}.\medskip

The Ehrenfest theorem for the quantum damped oscillator of Chru\'{s}ci\'{n}%
ski and Jurkowski coincides with our characteristic equation (\ref{CJcheq});
see Appendix~A for more details.\medskip

It is worth adding that in the momentum representation, when $%
p\leftrightarrow x,$ a rescaled Hamiltonian (\ref{CJHam}) ($\hslash =m\omega
_{0}=1$) takes the form%
\begin{equation}
H=\frac{\omega _{0}}{2}\left( \cosh ^{2}\left( \lambda t\right) \ p^{2}+%
\frac{x^{2}}{\cosh ^{2}\left( \lambda t\right) }\right) .  \label{CJhamP}
\end{equation}%
The corresponding characteristic equation%
\begin{equation}
\mu ^{\prime \prime }-2\lambda \tanh \left( \lambda t\right) \mu ^{\prime
}+\omega _{0}^{2}\mu =0  \label{CJcharP}
\end{equation}%
has a required elementary solution%
\begin{equation}
\mu =\frac{1}{\omega _{0}}\left( \lambda \cos \left( \omega t\right) \sinh
\left( \lambda t\right) +\omega \sin \left( \omega t\right) \cosh \left(
\lambda t\right) \right)  \label{CJmuP}
\end{equation}%
with $\mu ^{\prime }\left( 0\right) =2a\left( 0\right) =\omega _{0}$ and%
\begin{equation}
\mu \rightarrow \frac{1}{2\omega _{0}}e^{\lambda t}\left( \lambda \cos
\left( \omega t\right) +\omega \sin \left( \omega t\right) \right)
\label{CJmuasy}
\end{equation}%
as $t\rightarrow \infty .$ The Green function is given by formula (\ref{in2}%
) with the following coefficients:%
\begin{eqnarray}
\alpha \left( t\right) &=&\frac{\omega _{0}\cos \left( \omega t\right) }{%
2\cosh \left( \lambda t\right) \left( \lambda \cos \left( \omega t\right)
\sinh \left( \lambda t\right) +\omega \sin \left( \omega t\right) \cosh
\left( \lambda t\right) \right) },  \label{CJalphaP} \\
\beta \left( t\right) &=&-\frac{\omega _{0}}{\lambda \cos \left( \omega
t\right) \sinh \left( \lambda t\right) +\omega \sin \left( \omega t\right)
\cosh \left( \lambda t\right) },  \label{CJbetaP} \\
\gamma \left( t\right) &=&\frac{\omega _{0}\left( \omega \cos \left( \omega
t\right) \cosh \left( \lambda t\right) -\lambda \sin \left( \omega t\right)
\sinh \left( \lambda t\right) \right) }{2\omega \left( \lambda \cos \left(
\omega t\right) \sinh \left( \lambda t\right) +\omega \sin \left( \omega
t\right) \cosh \left( \lambda t\right) \right) }.  \label{CJgammaP}
\end{eqnarray}%
The details are left to the reader.

\subsection{A Modified Parametric Oscillator}

In a similar fashion we consider the following Hamiltonian:%
\begin{eqnarray}
&&H=\frac{\omega }{2}\left( \tanh ^{2}\left( \lambda t+\delta \right) \
p^{2}+\coth ^{2}\left( \lambda t+\delta \right) \ x^{2}\right)
\label{MPOHam} \\
&&\qquad +\frac{\lambda }{\sinh \left( 2\lambda t+2\delta \right) }\left(
px+xp\right) \qquad \left( \delta \neq 0\right) ,  \notag
\end{eqnarray}%
which seems to be missing in the available literature. The corresponding
characteristic equation:%
\begin{equation}
\mu ^{\prime \prime }-\frac{4\lambda }{\sinh \left( 2\lambda t+2\delta
\right) }\mu ^{\prime }+\left( \omega ^{2}+\frac{2\lambda ^{2}}{\sinh
^{2}\left( \lambda t+\delta \right) }\right) \mu =0  \label{MPOchar}
\end{equation}%
has an elementary solution of the form:%
\begin{equation}
\mu =\sin \left( \omega t\right) \frac{\tanh \left( \lambda t+\delta \right) 
}{\coth \delta }.  \label{MPOmu}
\end{equation}%
In the limit $t\rightarrow \infty ,$ $\mu \rightarrow \sin \left( \omega
t\right) \tanh \delta .$\medskip

The Green function can be found as follows%
\begin{equation}
G\left( x,y,t\right) =\sqrt{\frac{\coth \delta }{2\pi i\sin \left( \omega
t\right) \tanh \left( \lambda t+\delta \right) }}\ e^{i\left( \alpha \left(
t\right) x^{2}+\beta \left( t\right) xy+\gamma \left( t\right) y^{2}\right)
},  \label{MPOGreen}
\end{equation}%
where%
\begin{equation}
\alpha \left( t\right) =\frac{1}{2}\cot \left( \omega t\right) \coth
^{2}\left( \lambda t+\delta \right) ,  \label{MPOalpha}
\end{equation}%
\begin{equation}
\beta \left( t\right) =-\frac{\coth \delta }{\sin \left( \omega t\right) }%
\coth \left( \lambda t+\delta \right) ,  \label{MPObeta}
\end{equation}%
\begin{equation}
\gamma \left( t\right) =\frac{1}{2}\cot \left( \omega t\right) \coth
^{2}\delta .  \label{MPOgamma}
\end{equation}%
The Ehrenfest theorem coincides with the characteristic equation (\ref%
{MPOchar}). One should interchange $a\leftrightarrow b$ and $d\rightarrow -d$
in the momentum representation \cite{Cor-Sot:Sus}. The corresponding
solutions can be found with the help of the substitution $\delta \rightarrow
\delta +i\pi /2.$ The trigonometric cases, when $\lambda \rightarrow
i\lambda ,$ $\delta \rightarrow i\delta $ and $\omega \rightarrow -\omega ,$
are left to the reader.

\subsection{Parametric Oscillators}

In conclusion a somewhat related quantum parametric oscillator:%
\begin{equation}
H=\frac{1}{2}\left( p^{2}+\left( \omega ^{2}+\frac{2\lambda ^{2}}{\cosh
^{2}\left( \lambda t\right) }\right) x^{2}\right) ,  \label{ParamHam}
\end{equation}%
when%
\begin{equation}
\mu ^{\prime \prime }+\left( \omega ^{2}+\frac{2\lambda ^{2}}{\cosh
^{2}\left( \lambda t\right) }\right) \mu =0  \label{ParamChar}
\end{equation}%
and%
\begin{equation}
\mu =\frac{\lambda \cos \left( \omega t\right) \sinh \left( \lambda t\right)
+\omega \sin \left( \omega t\right) \cosh \left( \lambda t\right) }{\left(
\omega ^{2}+\lambda ^{2}\right) \cosh \left( \lambda t\right) },
\label{ParamMu}
\end{equation}%
has the Green function (\ref{in2}) with the following coefficients:%
\begin{eqnarray}
\alpha \left( t\right) &=&\frac{\left( \omega ^{2}+\lambda ^{2}\cosh
^{-2}\left( \lambda t\right) \right) \cos \left( \omega t\right) -\lambda
\omega \tanh \left( \lambda t\right) \sin \left( \omega t\right) }{2\left(
\omega \sin \left( \omega t\right) +\lambda \tanh \left( \lambda t\right)
\cos \left( \omega t\right) \right) },  \label{ParamAlpha} \\
\beta \left( t\right) &=&-\frac{\omega ^{2}+\lambda ^{2}}{\omega \sin \left(
\omega t\right) +\lambda \tanh \left( \lambda t\right) \cos \left( \omega
t\right) },  \label{ParamBeta} \\
\gamma \left( t\right) &=&\frac{\left( \omega ^{2}+\lambda ^{2}\right)
\left( \omega \cos \left( \omega t\right) -\lambda \tanh \left( \lambda
t\right) \sin \left( \omega t\right) \right) }{2\omega \left( \omega \sin
\left( \omega t\right) +\lambda \tanh \left( \lambda t\right) \cos \left(
\omega t\right) \right) }.  \label{ParamGamma}
\end{eqnarray}%
The Green function for the parametric oscillator in general:%
\begin{equation}
H=\frac{1}{2}\left( p^{2}+\omega ^{2}\left( t\right) x^{2}\right)
\label{ParametHam}
\end{equation}%
can be found, for example, in Ref.~\cite{Lan:Sus}. (The time-dependent
quantum oscillator was thoroughly examined by Husimi \cite{HusimiI53}, \cite%
{HusimiII53} and later many authors had treated different aspects of the
problem; see \cite{Dodonov:Man'koFIAN87}, \cite{Malkin:Man'ko79}, \cite%
{Malk:Man:Trif70}, \cite{Malk:Man:Trif73}, \cite{Perelomov:Popov69}, \cite%
{Per:Zel}, \cite{Popov:PerelomovI69}, \cite{Popov:PerelomovII69}, \cite%
{Schuch08} and \cite{Solimenoetal69}; a detailed bibliography is given in
Ref.~\cite{Camizetall71}.)

\section{Expectation Values of Quadratic Operators}

We start from a convenient differentiation formula.

\begin{lemma}
Let%
\begin{equation}
H=a\left( t\right) p^{2}+b\left( t\right) x^{2}+d\left( t\right) \left(
px+xp\right) ,  \label{genham}
\end{equation}%
\begin{equation}
O=A\left( t\right) p^{2}+B\left( t\right) x^{2}+C\left( t\right) \left(
px+xp\right)  \label{genop}
\end{equation}%
and%
\begin{equation}
\left\langle O\right\rangle =\left\langle \psi ,O\psi \right\rangle
=\int_{-\infty }^{\infty }\psi ^{\ast }O\psi \ dx,\qquad i\frac{\partial
\psi }{\partial t}=H\psi  \label{expect}
\end{equation}%
(we use the star for complex conjugate). Then%
\begin{eqnarray}
\frac{d}{dt}\left\langle O\right\rangle &=&\left( \frac{dA}{dt}+4\left(
aC-dA\right) \right) \left\langle p^{2}\right\rangle  \label{diffexp} \\
&&+\left( \frac{dB}{dt}+4\left( dB-bC\right) \right) \left\langle
x^{2}\right\rangle  \notag \\
&&+\left( \frac{dC}{dt}+2\left( aB-bA\right) \right) \left\langle
px+xp\right\rangle .  \notag
\end{eqnarray}
\end{lemma}

\begin{proof}
The time derivative of the expectation value can be written as \cite{La:Lif}%
, \cite{Merz}, \cite{Schiff}:%
\begin{equation}
\frac{d}{dt}\left\langle O\right\rangle =\left\langle \frac{\partial O}{%
\partial t}\right\rangle +\frac{1}{i}\left\langle \left[ O,H\right]
\right\rangle ,  \label{diffops}
\end{equation}%
where $\left[ O,H\right] =OH-HO$ (we freely interchange differentiation and
integration throughout the paper, it can be justified for certain classes of
solutions \cite{Lieb:Loss}, \cite{Oh89}, \cite{Per-G:Tor:Mont}, \cite{Velo96}%
). One should make use of the standard commutator properties, including
familiar identities%
\begin{eqnarray}
&&\left[ x^{2},p^{2}\right] =2i\left( px+xp\right) ,\qquad \left[ x,p^{2}%
\right] =2ip,\qquad \left[ x^{2},p\right] =2ix,  \label{commuts} \\
&&\left[ px+xp,p^{2}\right] =4ip^{2},\qquad \quad \left[ x^{2},px+xp\right]
=4ix^{2},  \notag
\end{eqnarray}%
in order to complete the proof.
\end{proof}

Quantum systems with the self-adjoint Hamiltonians (\ref{genham}) are called
the generalized harmonic oscillators \cite{Berry85}, \cite{Dod:Mal:Man75}, 
\cite{Hannay85}, \cite{Leach90}, \cite{Wolf81}, \cite%
{Yeon:Lee:Um:George:Pandey93}. At the same time one has to deal with
non-self-adjoint Hamiltonians in the theory of dissipative quantum systems
(see, for example, \cite{Cor-Sot:Sua:Sus}, \cite{Dekker81}, \cite{Kochan10}, 
\cite{Tarasov01}, \cite{Um:Yeon:George} and references therein) or when
using separation of variables in an accelerating frame of reference for a
charged particle moving in an uniform time-dependent magnetic field \cite%
{Cor-Sot:Lop:Sua:Sus}. An extension to the case of non-self-adjoint
Hamiltonians is as follows.

\begin{lemma}
If%
\begin{equation}
H=a\left( t\right) p^{2}+b\left( t\right) x^{2}+c\left( t\right) px+d\left(
t\right) xp,  \label{GenHam}
\end{equation}%
\begin{equation}
O=A\left( t\right) p^{2}+B\left( t\right) x^{2}+C\left( t\right) px+D\left(
t\right) xp,  \label{GenOp}
\end{equation}%
then%
\begin{eqnarray}
\frac{d}{dt}\left\langle O\right\rangle &=&\left( \frac{dA}{dt}+2a\left(
C+D\right) -\left( 3c+d\right) A\right) \left\langle p^{2}\right\rangle
\label{GenSys} \\
&&+\left( \frac{dB}{dt}-2b\left( C+D\right) +\left( c+3d\right) B\right)
\left\langle x^{2}\right\rangle  \notag \\
&&+\left( \frac{dC}{dt}+2\left( aB-bA\right) -\left( c-d\right) C\right)
\left\langle px\right\rangle  \notag \\
&&+\left( \frac{dD}{dt}+2\left( aB-bA\right) -\left( c-d\right) D\right)
\left\langle xp\right\rangle .  \notag
\end{eqnarray}
\end{lemma}

\begin{proof}
One should use%
\begin{equation}
\frac{d}{dt}\left\langle O\right\rangle =\left\langle \frac{\partial O}{%
\partial t}\right\rangle +\frac{1}{i}\left\langle OH-H^{\dagger
}O\right\rangle ,  \label{DiffOper}
\end{equation}%
where $H^{\dagger }$ is the Hermitian adjoint of the Hamiltonian operator $%
H. $ Our formula is a simple extension of the well-known expression \cite%
{La:Lif}, \cite{Merz}, \cite{Schiff} to the case of a non-self-adjoint
Hamiltonian \cite{Cor-Sot:Sua:Sus}. Standard commutator evaluations complete
the proof.\medskip
\end{proof}

Polynomial operators of the higher orders in $x$ and $p$ can be
differentiated in a similar fashion. An analog of the product rule is given
in \cite{Suslov10}. The details are left to the reader.

\section{Energy Operators and Quadratic Invariants}

In the case of the time-independent Hamiltonian, one gets%
\begin{equation}
\frac{d}{dt}\left\langle H\right\rangle =0  \label{constant}
\end{equation}%
by (\ref{diffops}). The law of conservation of energy states that%
\begin{equation}
E=\left\langle H\right\rangle =constant.  \label{constantenergy}
\end{equation}%
In general one has to construct quantum integrals of motion, or dynamical
invariants, that are different from the variable Hamiltonian (see, for
example, \cite{Lewis:Riesen69}, \cite{Wolf81}, \cite%
{Yeon:Lee:Um:George:Pandey93}; linear case is dealt with in \cite{Dod:Man79}%
, \cite{Dodonov:Man'koFIAN87}, \cite{Malkin:Man'ko79}, \cite{Malk:Man:Trif73}
and Appendix~C).

\subsection{Energy Operators}

A familiar definition is in order (see, for example, \cite%
{Dodonov:Man'koFIAN87}, \cite{Malkin:Man'ko79}).

\begin{definition}
We call the quadratic operator \textrm{(\ref{genop})} an energy operator $E,$
or a quadratic (dynamical) invariant, if%
\begin{equation}
\frac{d}{dt}\left\langle E\right\rangle =0  \label{defenergy}
\end{equation}%
for the corresponding variable Hamiltonian \textrm{(\ref{genham})}.
\end{definition}

By Lemma~1 the coefficients of an energy operator,%
\begin{equation}
E=A\left( t\right) p^{2}+B\left( t\right) x^{2}+C\left( t\right) \left(
px+xp\right) ,  \label{energyop}
\end{equation}%
must satisfy the system of ordinary differential equations:%
\begin{eqnarray}
\frac{dA}{dt}+4\left( a\left( t\right) C-d\left( t\right) A\right) &=&0,
\label{energysysA} \\
\frac{dB}{dt}+4\left( d\left( t\right) B-b\left( t\right) C\right) &=&0,
\label{energysysB} \\
\frac{dC}{dt}+2\left( a\left( t\right) B-b\left( t\right) A\right) &=&0.
\label{energysysC}
\end{eqnarray}%
In general a unique solution of this system with respect to arbitrary
initial conditions $A_{0}=A\left( 0\right) ,$ $B_{0}=B\left( 0\right) ,$ $%
C_{0}=C\left( 0\right) $ \cite{HilleODE} determines a three-parameter family
of the quadratic invariants (\ref{energyop}). Special cases, when solutions
can be found explicitly, are of the most practical importance.\medskip

In this section we find the simplest energy operators for all quadratic
models under consideration as follows:%
\begin{equation}
E=\frac{\omega _{0}}{2}\left( e^{-2\lambda t}\ p^{2}+e^{2\lambda t}\
x^{2}\right) +\frac{\lambda }{2}\left( px+xp\right) ,  \label{EO1}
\end{equation}%
\begin{equation}
E=\frac{\omega _{0}}{2}\left( e^{-2\lambda t}\ p^{2}+e^{2\lambda t}\
x^{2}\right) -\frac{\lambda }{2}\left( px+xp\right) ,  \label{EO2}
\end{equation}%
\begin{equation}
E=\frac{1}{2}\cos 2t\ \left( p^{2}-x^{2}\right) +\frac{1}{2}\sin 2t\ \left(
px+px\right) ,  \label{EO3}
\end{equation}%
\begin{equation}
E=\tanh ^{2}\left( \lambda t+\delta \right) \ p^{2}+\coth ^{2}\left( \lambda
t+\delta \right) \ x^{2}  \label{EQ4}
\end{equation}%
for the Caldirola-Kanai Hamiltonian (\ref{CKham}) \cite{Svin75}, the
modified Caldirola-Kanai Hamiltonian (\ref{modCKham}), the modified
oscillator of Meiler, Cordero-Soto and Suslov (\ref{mod1}) and for the
modified parametric oscillator (\ref{MPOHam}), respectively. Their
coefficients solve the corresponding systems (\ref{energysysA})--(\ref%
{energysysC}) for special initial data.\medskip

An energy operator for the united model (\ref{UMHam}) is given by%
\begin{equation}
E=\frac{\omega _{0}}{2}e^{\mu t}\left( e^{-2\lambda t}\ p^{2}+e^{2\lambda
t}\ x^{2}\right) +\frac{1}{2}\left( \lambda -\mu \right) e^{\mu t}\left(
px+xp\right) .  \label{EOUM}
\end{equation}%
One should use Lemma~2; verification is left to the reader. Finally an
energy operator for the quantum damped oscillator of Chru\'{s}ci\'{n}ski and
Jurkowski with a rescaled Hamiltonian (\ref{CJham}) is given by expression (%
\ref{CJEnergy}). A general case of the variable quadratic Hamiltonian is
discussed in Theorem~1.

\subsection{ The Lewis--Riesenfeld Invariant}

Classical Hamiltonian of the generalized harmonic oscillator can be
transformed into the Hamiltonian of a parametric oscillator \cite{Berry85}, 
\cite{Hannay85}, \cite{PadillaMaster}, \cite{Yeon:Lee:Um:George:Pandey93}.
All quadratic invariants of the quantum parametric oscillator (\ref%
{ParametHam}) can be found as follows \cite{Lewis67}, \cite{Lewis68}, \cite%
{Lewis68a}, \cite{Lewis:Riesen69}. The corresponding system,%
\begin{eqnarray}
&&A^{\prime }+2C=0,  \label{ParamOscA} \\
&&B^{\prime }-2\omega ^{2}\left( t\right) C=0,  \label{ParamOscB} \\
&&C^{\prime }+B-\omega ^{2}\left( t\right) A=0,  \label{ParamOscC}
\end{eqnarray}%
is integrated by the substitution $A=\kappa ^{2}.$ Then $C=-\kappa \kappa
^{\prime },$ $B=\kappa \kappa ^{\prime \prime }+\left( \kappa ^{\prime
}\right) ^{2}+\omega ^{2}\left( t\right) \kappa ^{2}$ and equation (\ref%
{ParamOscB}) becomes%
\begin{eqnarray*}
\left( \kappa \kappa ^{\prime \prime }+\left( \kappa ^{\prime }\right)
^{2}+\omega ^{2}\left( t\right) \kappa ^{2}\right) ^{\prime }+2\omega
^{2}\left( t\right) \kappa \kappa ^{\prime } &=&0, \\
\kappa \left( \kappa ^{\prime \prime }+\omega ^{2}\left( t\right) \kappa
\right) ^{\prime }+3\kappa ^{\prime }\left( \kappa ^{\prime \prime }+\omega
^{2}\left( t\right) \kappa \right) &=&0
\end{eqnarray*}%
or with an integrating factor:%
\begin{equation}
\frac{d}{dt}\left( \kappa ^{3}\left( \kappa ^{\prime \prime }+\omega
^{2}\left( t\right) \kappa \right) \right) =0  \label{ODEint}
\end{equation}%
(see \cite{Lewis:Riesen69} and \cite{Leach:Andriopo08}). Thus%
\begin{equation}
\kappa ^{\prime \prime }+\omega ^{2}\left( t\right) \kappa =\frac{c_{0}}{%
\kappa ^{3}}\qquad \left( c_{0}=0,1\right)  \label{NonlinearODE}
\end{equation}%
and a general solution of the system (\ref{ParamOscA})--(\ref{ParamOscC}) is
given by%
\begin{equation}
A=\kappa ^{2},\quad B=\left( \kappa ^{\prime }\right) ^{2}+\frac{c_{0}}{%
\kappa ^{2}},\quad C=-\kappa \kappa ^{\prime }  \label{ParamSysSol}
\end{equation}%
in terms of solutions of the nonlinear equation (\ref{NonlinearODE}), which
is called Ermakov's equation, when $c_{0}=1$ \cite{Ermakov} (see also, \cite%
{Leach:Andrio08}, \cite{Lewis68a}, \cite{Pinney50} and \cite{Schuch08}).
Thus the quadratic integrals of motion can be presented in the form \cite%
{Lewis:Riesen69}:%
\begin{equation}
E=\left( \kappa p-\kappa ^{\prime }x\right) ^{2}+\frac{c_{0}}{\kappa ^{2}}%
x^{2}  \label{ParamQuadInv}
\end{equation}%
for any given solution of the Ermakov equation (\ref{NonlinearODE}). This
quantum invariant is an analog of the Ermakov--Lewis integral of motion for
the classical parametric oscillator \cite{Ermakov}, \cite{Lewis67}, \cite%
{Lewis68}, \cite{Lewis68a}, \cite{Symon70}.\medskip

In general if two linearly independent solutions of the classical parametric
oscillator equation are available:%
\begin{equation}
u^{\prime \prime }+\omega ^{2}\left( t\right) u=0,\qquad v^{\prime \prime
}+\omega ^{2}\left( t\right) v=0,  \label{ParamOscillatorClSol}
\end{equation}%
then solutions of the nonlinear Ermakov equation:%
\begin{equation}
\kappa ^{\prime \prime }+\omega ^{2}\left( t\right) \kappa =\frac{1}{\kappa
^{3}}  \label{ErmakovEquationCl}
\end{equation}%
are given by%
\begin{equation}
\kappa =\left( Au^{2}+2Buv+Cv^{2}\right) ^{1/2}
\label{ErmakovEqPinneySolution}
\end{equation}%
(so-called Pinney's solution \cite{Pinney50}, \cite{Eliezer:Gray76}, \cite%
{Leach:Andrio08}, \cite{Lewis68a}, \cite{PadillaMaster}), where the
constants $A,$ $B$ and $C$ are related according to $AC-B^{2}=1/W^{2}$ with $%
W$ being the constant Wronskian of the two linearly independent
solutions.\medskip

For example, in the case of the simple harmonic oscillator with $\omega
\left( t\right) =1,$ there are two elementary solutions:%
\begin{equation}
\kappa =1\quad \left( c_{0}=1\right) ,\qquad \kappa =\cos t\quad \left(
c_{0}=0\right)  \label{ElemSol}
\end{equation}%
and the energy operators are given by%
\begin{eqnarray}
H &=&\frac{1}{2}\left( p^{2}+x^{2}\right) ,  \label{HarmHam} \\
E &=&\left( \cos t\ p+\sin t\ x\right) ^{2}.  \label{MC-SShamenergy}
\end{eqnarray}%
It provides a somewhat better understanding of the nature of the Hamiltonian
discussed by Meiler, Cordero-Soto and Suslov \cite{Me:Co:Su} --- this
operator plays a role of the simplest time-dependent quadratic integral of
motion for the linear harmonic oscillator.\medskip

In a similar fashion the dynamical invariants of the parametric oscillator (%
\ref{ParamHam}) are given by the expression (\ref{ParamQuadInv}) with $%
c_{0}\neq 0.$ In the Pinney solution (\ref{ErmakovEqPinneySolution}) one can
choose%
\begin{eqnarray}
u &=&\frac{\omega \cos \left( \omega t\right) \cosh \left( \lambda t\right)
-\lambda \sin \left( \omega t\right) \sinh \left( \lambda t\right) }{\cosh
\left( \lambda t\right) },  \label{ParamClSolU} \\
v &=&\frac{\omega \sin \left( \omega t\right) \cosh \left( \lambda t\right)
+\lambda \cos \left( \omega t\right) \sinh \left( \lambda t\right) }{\cosh
\left( \lambda t\right) }  \label{ParamClSolV}
\end{eqnarray}%
as two linearly independent solutions of the classical equation of motion (%
\ref{ParamChar}) with $W\left( u,v\right) =\omega \left( \omega ^{2}+\lambda
^{2}\right) .$ If $A=C$ and $B=0,$ then%
\begin{equation}
\kappa =\left( \omega ^{2}+\lambda ^{2}\tanh ^{2}\left( \lambda t\right)
\right) ^{1/2}  \label{ParamEnergySol}
\end{equation}%
is a particular solution of the corresponding Ermakov equation:%
\begin{equation}
\kappa ^{\prime \prime }+\left( \omega ^{2}+\frac{2\lambda ^{2}}{\cosh
^{2}\left( \lambda t\right) }\right) \kappa =\frac{\omega ^{2}\left( \lambda
^{2}+\omega ^{2}\right) ^{2}}{\kappa ^{3}}.  \label{ParamErmakovEquation}
\end{equation}%
The simplest positive energy integral for our parametric oscillator (\ref%
{ParamHam}) is given by%
\begin{eqnarray}
E &=&\left( \omega ^{2}+\lambda ^{2}\tanh ^{2}\left( \lambda t\right)
\right) \ p^{2}+\lambda ^{3}\frac{\sinh \left( \lambda t\right) }{\cosh
^{3}\left( \lambda t\right) }\ \left( px+xp\right)  \label{ParamEnergy} \\
&&+\frac{\lambda ^{6}\sinh ^{2}\left( \lambda t\right) +\omega ^{2}\left(
\lambda ^{2}+\omega ^{2}\right) ^{2}\cosh ^{6}\left( \lambda t\right) }{%
\cosh ^{6}\left( \lambda t\right) \left( \omega ^{2}+\lambda ^{2}\tanh
^{2}\left( \lambda t\right) \right) }\ x^{2}.  \notag
\end{eqnarray}%
Another possibility is to take a general solution of (\ref{ParamChar}) with $%
c_{0}=0.$

\subsection{An Extension to General Quadratic Hamiltonians}

We consider the following generalization of the Lewis--Riesenfeld invariant (%
\ref{ParamQuadInv}) (see also \cite{Leach90}, \cite%
{Yeon:Lee:Um:George:Pandey93}).

\begin{theorem}
The dynamical invariants for the general quadratic Hamiltonian (\ref{GenHam}%
) are given by%
\begin{equation}
E=\frac{1}{\mu _{1}}\left( \kappa \ p-\frac{1}{2a}\frac{d\kappa }{dt}\
x\right) ^{2}+\frac{C_{0}}{\mu _{2}\kappa ^{2}}\ x^{2},  \label{GenQuadInv}
\end{equation}%
where $C_{0}$ is a constant,%
\begin{equation}
\mu _{1}=\exp \left( -\int_{0}^{t}\left( 3c+d\right) \ ds\right) ,\quad \mu
_{2}=\exp \left( \int_{0}^{t}\left( c+3d\right) \ ds\right) ,
\label{IntFacts}
\end{equation}%
and $\kappa $ satisfies the auxiliary nonlinear equation:%
\begin{equation}
k\frac{d}{dt}\left( k\frac{d\kappa }{dt}\right) +4abk^{2}\kappa =\frac{C_{0}%
}{\kappa ^{3}},  \label{AuxEq}
\end{equation}%
where%
\begin{equation}
k=\frac{1}{2a}\exp \left( 2\int_{0}^{t}\left( c+d\right) \ ds\right) .
\label{Key}
\end{equation}
(For the self-adjoint Hamiltonians $c=d.$)
\end{theorem}

The case, $a=1/2,$ $b=\omega ^{2}\left( t\right) /2$ and $c=d=0,$
corresponds to the original invariant (\ref{ParamQuadInv}).

\begin{proof}
By Lemma~2 in order to find quadratic invariants of the form%
\begin{equation}
E=Ap^{2}+Bx^{2}+Cpx+Dxp  \label{Energy}
\end{equation}%
we have to solve the following system of ordinary differential equations:%
\begin{eqnarray}
\frac{dA}{dt}+2a\left( C+D\right) -\left( 3c+d\right) A &=&0,  \label{EquatA}
\\
\frac{dB}{dt}-2b\left( C+D\right) +\left( c+3d\right) B &=&0,  \label{EquatB}
\\
\frac{dC}{dt}+2\left( aB-bA\right) -\left( c-d\right) C &=&0,  \label{EquatC}
\\
\frac{dD}{dt}+2\left( aB-bA\right) -\left( c-d\right) D &=&0,  \label{EquatD}
\end{eqnarray}%
say, for arbitrary analytic coefficients $a\left( t\right) ,$ $b\left(
t\right) ,$ $c\left( t\right) $ and $d\left( t\right) .$ The substitution $%
C=C_{1}+D_{1},$ $D=C_{1}-D_{1}$ allows one to transform the last two
equations:%
\begin{eqnarray}
&&\frac{dC_{1}}{dt}+2\left( aB-bA\right) -\left( c-d\right) C_{1}=0,
\label{EquatC1} \\
&&\frac{dD_{1}}{dt}=\left( c-d\right) D_{1},\quad D_{1}=\text{constant\ }%
\exp \left( \int_{0}^{t}\left( c-d\right) \ ds\right) .  \label{EquatD1}
\end{eqnarray}%
Then%
\begin{equation*}
Cpx+Dxp=C_{1}\left( px+xp\right) +D_{1}\left( px-xp\right)
\end{equation*}%
and, in view of the canonical commutation relation, the coefficient $D_{1}$
can be eliminated from the consideration as belonging to the linear
invariants (see appendix~C).

Introducing integrating factors into (\ref{EquatA}), (\ref{EquatB}) and (\ref%
{EquatC1}), we get%
\begin{eqnarray}
&&\frac{d}{dt}\left( \mu _{1}A\right) +4a\mu _{1}C_{1}=0,\qquad \frac{\mu
_{1}^{\prime }}{\mu _{1}}=-3c-d, \\
&&\frac{d}{dt}\left( \mu _{2}B\right) -4b\mu _{2}C_{1}=0,\qquad \frac{\mu
_{2}^{\prime }}{\mu _{2}}=c+3d, \\
&&\frac{d}{dt}\left( \mu _{3}C_{1}\right) +2\mu _{3}\left( aB-bA\right)
=0,\qquad \frac{\mu _{3}^{\prime }}{\mu _{3}}=-c+d
\end{eqnarray}%
with $\mu _{3}^{2}=\mu _{1}\mu _{2}.$ After the substitution%
\begin{equation}
\widetilde{A}=\mu _{1}A,\qquad \widetilde{B}=\mu _{2}B,\qquad \widetilde{C}%
=\mu _{3}C_{1},  \label{SubTilde}
\end{equation}%
the system takes the form%
\begin{eqnarray}
&&\frac{d\widetilde{A}}{dt}+4a\sqrt{\frac{\mu _{1}}{\mu _{2}}}\ \widetilde{C}%
=0, \\
&&\frac{d\widetilde{B}}{dt}-4b\sqrt{\frac{\mu _{2}}{\mu _{1}}}\ \widetilde{C}%
=0, \\
&&\frac{d\widetilde{C}}{dt}+2\left( a\sqrt{\frac{\mu _{1}}{\mu _{2}}}\ 
\widetilde{B}-b\sqrt{\frac{\mu _{2}}{\mu _{1}}}\ \widetilde{A}\right) =0.
\end{eqnarray}%
Introducing a \textquotedblleft proper time\textquotedblright :%
\begin{equation}
\tau =\int_{0}^{t}2a\sqrt{\frac{\mu _{1}}{\mu _{2}}}\ ds,  \label{PropTime}
\end{equation}%
we finally obtain:%
\begin{eqnarray}
&&\frac{d\widetilde{A}}{d\tau }+2\widetilde{C}=0,  \label{LRA} \\
&&\frac{d\widetilde{B}}{d\tau }-2\omega ^{2}\left( \tau \right) \widetilde{C}%
=0,  \label{LRB} \\
&&\frac{d\widetilde{C}}{d\tau }+\widetilde{B}-\omega ^{2}\left( \tau \right) 
\widetilde{A}=0,\quad \omega ^{2}\left( \tau \right) =\frac{b\mu _{2}}{a\mu
_{1}},  \label{LRC}
\end{eqnarray}%
which is identical to the original Lewis--Riesenfeld system (\ref{ParamOscA}%
)--(\ref{ParamOscC}) (positivity of $\omega ^{2}$ is not required). The
solution is given by%
\begin{equation}
\widetilde{A}=\kappa ^{2},\quad \widetilde{B}=\left( \frac{d\kappa }{d\tau }%
\right) ^{2}+\frac{C_{0}}{\kappa ^{2}},\quad \widetilde{C}=-\kappa \frac{%
d\kappa }{d\tau },  \label{LRSysSol}
\end{equation}%
where $\kappa $ satisfies the Ermakov equation:%
\begin{equation}
\frac{d^{2}\kappa }{d\tau ^{2}}+\omega ^{2}\left( \tau \right) \kappa =\frac{%
C_{0}}{\kappa ^{3}},\quad \omega ^{2}\left( \tau \right) =\frac{b\mu _{2}}{%
a\mu _{1}},  \label{ErmakovEquation}
\end{equation}%
with respect to the new time (\ref{PropTime}). In view of%
\begin{equation}
\frac{d}{d\tau }=k\frac{d}{dt},\qquad k=\frac{1}{2a}\exp \left(
2\int_{0}^{t}\left( c+d\right) \ ds\right) ,  \label{DiffTau}
\end{equation}%
the Ermakov equation (\ref{ErmakovEquation}) is transformed into our
auxiliary equation (\ref{AuxEq}). A back substitution results in the
dynamical invariant (\ref{GenQuadInv}) when the square is completed.
\end{proof}

\begin{lemma}
The dynamical invariant (\ref{GenQuadInv}) can be represented in more
symmetric form%
\begin{eqnarray}
E &=&\left( \left( \mu \ p-\frac{1}{2a}\left( \frac{d\mu }{dt}-\left(
c+d\right) \mu \right) \ x\right) ^{2}+\frac{C_{0}}{\mu ^{2}}\ x^{2}\right)
\label{InvSymmForm} \\
&&\times \exp \left( \int_{0}^{t}\left( c-d\right) \ ds\right) ,  \notag
\end{eqnarray}%
where $C_{0}$ is a constant and $\mu $ is a solution of the following
auxiliary equation:%
\begin{equation}
\mu ^{\prime \prime }-\frac{a^{\prime }}{a}\mu ^{\prime }+\left( 4ab+\left( 
\frac{a^{\prime }}{a}-c-d\right) \left( c+d\right) -c^{\prime }-d^{\prime
}\right) \mu =C_{0}\frac{\left( 2a\right) ^{2}}{\mu ^{3}}.
\label{AuxEquation}
\end{equation}
\end{lemma}

\begin{proof}
Use the substitution%
\begin{equation}
\kappa =\mu \exp \left( -\int_{0}^{t}\left( c+d\right) \ ds\right)
\label{MuSubstitution}
\end{equation}%
in (\ref{GenQuadInv}) and (\ref{AuxEq}). A somewhat different proof is given
in \cite{Suslov10}.
\end{proof}

The corresponding classical invariant is discussed, for example, in Refs.~%
\cite{Symon70} and \cite{Yeon:Lee:Um:George:Pandey93}. (Compare also our
expression (\ref{InvSymmForm}) with the one given in the last paper for the
self-adjoint case; we give a detailed proof for the non-self-adjoint
Hamiltonians and emphasize connection with the Ermakov equation.)\medskip

It is worth noting, in conclusion, that, if $\mu _{1}$ and $\mu _{2}$ are
two linearly independent solutions of the linear equation:%
\begin{equation}
\mu ^{\prime \prime }-\frac{a^{\prime }}{a}\mu ^{\prime }+\left( 4ab+\left( 
\frac{a^{\prime }}{a}-c-d\right) \left( c+d\right) -c^{\prime }-d^{\prime
}\right) \mu =0,  \label{LinEquation}
\end{equation}%
the general solution of the nonlinear auxiliary equation (\ref{AuxEquation})
is given by%
\begin{equation}
\mu =\left( A\mu _{1}^{2}+2B\mu _{1}\mu _{2}+C\mu _{2}^{2}\right) ^{1/2},
\label{SolNonLinEquation}
\end{equation}%
where the constants $A,$ $B$ and $C$ are related according to%
\begin{equation}
AC-B^{2}=C_{0}\frac{\left( 2a\right) ^{2}}{W^{2}\left( \mu _{1},\mu
_{2}\right) }  \label{NonLinWronskian}
\end{equation}%
with $W\left( \mu _{1},\mu _{2}\right) =\mu _{1}\mu _{2}^{\prime }-\mu
_{1}^{\prime }\mu _{2}=constant\ \left( 2a\right) $ being the Wronskian of
the two linearly independent solutions. This is a simple extension of
Pinney's solution (\ref{ErmakovEqPinneySolution}); our equations (\ref%
{AuxEquation}) and (\ref{LinEquation}) form the generalized Ermakov system 
\cite{Eliezer:Gray76}, \cite{PadillaMaster}. Further generalization of the
superposition formula (\ref{SolNonLinEquation})--(\ref{NonLinWronskian}) is
discussed in Ref.~\cite{Suslov10}. (If $C_{0}\neq 0,$ the substitution $\mu
\rightarrow $ $C_{0}^{1/4}\mu $ reduces equation (\ref{AuxEquation}) to a
similar form with $C_{0}=1.)$ Special case of the time-dependent damped
harmonic oscillator, when $a=e^{-F\left( t\right) }/2,$ $b=\omega ^{2}\left(
t\right) e^{F\left( t\right) }/2,$ $F\left( t\right) =\int_{0}^{t}f\left(
s\right) \ ds$ and $c=d=0,$ is discussed in \cite{LeachAmJPhys78}, \cite%
{LeachSIAM78}.

\subsection{An Example}

The simplest energy operators have been already discussed in section~4.1 for
all models of quantum oscillators under consideration. In order to
demonstrate how the general approach works we discuss the united Hamiltonian
(\ref{UMHam}), when $a=\left( \omega _{0}/2\right) e^{-2\lambda t},$ $%
b=\left( \omega _{0}/2\right) e^{2\lambda t}$ and $c=0,$ $d=-\mu .$ A direct
calculation shows that the function%
\begin{equation}
\kappa =\sqrt{\frac{\omega _{0}}{2}}e^{-\lambda t}  \label{UMkappa}
\end{equation}%
satisfies the following equation%
\begin{equation}
\kappa ^{\prime \prime }+2\lambda \kappa ^{\prime }+\omega _{0}^{2}\kappa
=\left( \frac{\omega _{0}\omega }{2}\right) ^{2}\frac{e^{-4\lambda t}}{%
\kappa ^{3}},\quad \omega ^{2}=\omega _{0}^{2}-\left( \lambda -\mu \right)
^{2}>0,  \label{UMAuxEq}
\end{equation}%
which corresponds to the nonlinear auxiliary equation (\ref{AuxEquation})
with $C_{0}=\omega ^{2}/4.$ The quadratic invariant (\ref{InvSymmForm})
simplifies to the previously found expression (\ref{EOUM}). Solution (\ref%
{SolNonLinEquation}) can be used for the most general case. Details are left
to the reader.

\subsection{Factorization of the Dynamical Invariant}

Following Ref.~\cite{Cor-Sot:Sua:Sus} the energy operator (\ref{InvSymmForm}%
) can be presented in the standard harmonic oscillator form: 
\begin{equation}
E=\frac{\omega \left( t\right) }{2}\left( \widehat{a}\left( t\right) 
\widehat{a}^{\dagger }\left( t\right) +\widehat{a}^{\dagger }\left( t\right) 
\widehat{a}\left( t\right) \right) ,  \label{EnOperFactor}
\end{equation}%
where%
\begin{equation}
\omega \left( t\right) =\omega _{0}\exp \left( \int_{0}^{t}\left( c-d\right)
\ ds\right) ,\qquad \omega _{0}=2\sqrt{C_{0}}>0,  \label{omega(t)}
\end{equation}%
\begin{eqnarray}
\widehat{a}\left( t\right) &=&\left( \frac{\sqrt{\omega _{0}}}{2\mu }-i\frac{%
\mu ^{\prime }-\left( c+d\right) \mu }{2a\sqrt{\omega _{0}}}\right) x+\frac{%
\mu }{\sqrt{\omega _{0}}}\frac{\partial }{\partial x},  \label{a(t)} \\
\widehat{a}^{\dagger }\left( t\right) &=&\left( \frac{\sqrt{\omega _{0}}}{%
2\mu }+i\frac{\mu ^{\prime }-\left( c+d\right) \mu }{2a\sqrt{\omega _{0}}}%
\right) x-\frac{\mu }{\sqrt{\omega _{0}}}\frac{\partial }{\partial x},
\label{across(t)}
\end{eqnarray}%
and $\mu $ is a solution of the nonlinear auxiliary equation (\ref%
{AuxEquation}). Here the time-dependent annihilation $\widehat{a}\left(
t\right) $ and creation $\widehat{a}^{\dagger }\left( t\right) $ operators
satisfy the usual commutation relation:%
\begin{equation}
\widehat{a}\left( t\right) \widehat{a}^{\dagger }\left( t\right) -\widehat{a}%
^{\dagger }\left( t\right) \widehat{a}\left( t\right) =1.
\label{commutatora(t)across(t)}
\end{equation}%
The oscillator-type spectrum and the corresponding time-dependent
eigenfunctions of the dynamical invariant $E$ can be obtain now in a
standard way by using the Heisenberg--Weyl algebra of the rasing and
lowering operators (a \textquotedblleft second
quantization\textquotedblright\ \cite{Lewis:Riesen69}, the Fock states).
Explicit solution of the Cauchy initial value problem in terms of the
quadratic invariant eigenfunction expansion is found in Ref.~\cite{Suslov10}%
. In addition the $n$-dimensional oscillator wave functions form a basis of
the irreducible unitary representation of the Lie algebra of the noncompact
group $SU\left( 1,1\right) $ corresponding to the discrete positive series $%
\mathcal{D}_{+}^{j}$ (see \cite{Me:Co:Su}, \cite{Ni:Su:Uv} and \cite%
{Smir:Shit}).\smallskip\ Our operators (\ref{a(t)})--(\ref{across(t)}) allow
one to extend these group-theoretical properties to the general dynamical
invariant (\ref{EnOperFactor}). We shall further elaborate on these
connections elsewhere.

\section{Application to the Cauchy Initial Value Problems}

Explicit solution of the initial value problem in terms of eigenfunctions of
the general quadratic invariant is given in Ref.~\cite{Suslov10}. Here we
formulate the following uniqueness result.

\begin{lemma}
Suppose that the expectation value%
\begin{equation}
\left\langle H_{0}\right\rangle =\left\langle \psi ,H_{0}\psi \right\rangle
\geq 0  \label{exppos}
\end{equation}%
for a positive quadratic operator%
\begin{equation}
H_{0}=f\left( t\right) \left( \alpha \left( t\right) p+\beta \left( t\right)
x\right) ^{2}+g\left( t\right) x^{2}\qquad \left( f\left( t\right) \geq 0,\
g\left( t\right) >0\right)  \label{posop}
\end{equation}%
($\alpha \left( t\right) $ and $\beta \left( t\right) $ are real-valued
functions) vanishes for all $t\in \lbrack 0,T):$%
\begin{equation}
\left\langle H_{0}\right\rangle =\left\langle H_{0}\right\rangle \left(
t\right) =\left\langle H_{0}\right\rangle \left( 0\right) =0,  \label{indata}
\end{equation}%
when $\psi \left( x,0\right) =0$ almost everywhere. Then the corresponding
Cauchy initial value problem%
\begin{equation}
i\frac{\partial \psi }{\partial t}=H\psi ,\qquad \psi \left( x,0\right)
=\varphi \left( x\right)  \label{Cauchyivp}
\end{equation}%
may have only one solution, when $x\psi \left( x,t\right) \in L^{2}\left( 
\mathbb{R}
\right) $ (if $H_{0}=g\left( t\right) I,$ where $I=id$ is the identity
operator, $\psi \in L^{2}\left( 
\mathbb{R}
\right) $).
\end{lemma}

Here it is not assumed that $H_{0}$ is the quantum integral of motion when $%
\frac{d}{dt}\left\langle H_{0}\right\rangle \equiv 0.$

\begin{proof}
If there are two solutions:%
\begin{equation*}
i\frac{\partial \psi _{1}}{\partial t}=H\psi _{1},\qquad i\frac{\partial
\psi _{2}}{\partial t}=H\psi _{2}
\end{equation*}%
with the same initial condition $\psi _{1}\left( x,0\right) =\psi _{2}\left(
x,0\right) =\varphi \left( x\right) ,$ then by the superposition principle
the function $\psi =\psi _{1}-\psi _{2}$ is also a solution with respect to
the zero initial data $\psi \left( x,0\right) =\varphi \left( x\right)
-\varphi \left( x\right) =0.$ By the hypothesis of the lemma%
\begin{equation*}
\left\langle \psi ,H_{0}\psi \right\rangle =f\left( t\right) \left\langle
\left( \alpha p+\beta x\right) \psi ,\left( \alpha p+\beta x\right) \psi
\right\rangle +g\left( t\right) \left\langle x\psi ,x\psi \right\rangle =0
\end{equation*}%
for all $t\in \lbrack 0,T).$ Therefore $x\psi \left( x,t\right) =x\left(
\psi _{1}\left( x,t\right) -\psi _{2}\left( x,t\right) \right) =0$ and $\psi
_{1}\left( x,t\right) =\psi _{2}\left( x,t\right) $ almost everywhere for
all $t>0$ by the axiom of the inner product in $L^{2}\left( 
\mathbb{R}
\right) .$
\end{proof}

In order to apply this lemma to the variable Hamiltonians one has to
identify the corresponding positive operators $H_{0}$ and establish their
required uniqueness dynamics properties with respect to the zero initial
data. In addition to the simplest available dynamical invariant (\ref{b1}),
it is worth exploring other (quadratic) possibilities. The authors believe
that it is interesting and may be important on its own. For example, our
approach gives an opportunity to determine a complete time-evolution of the
standard deviations (\ref{bdp})--(\ref{bdx}) for each of the generalized
harmonic oscillators under consideration. The details will be discussed
elsewhere.

\subsection{The Caldirola-Kanai Hamiltonian}

The required operators are given by%
\begin{equation}
H=H_{0}=\frac{\omega _{0}}{2}\left( e^{-2\lambda t}\ p^{2}+e^{2\lambda t}\
x^{2}\right) ,  \label{CKH1}
\end{equation}%
\begin{equation}
L=\frac{\partial H}{\partial t}=\lambda \omega _{0}\left( -e^{-2\lambda t}\
p^{2}+e^{2\lambda t}\ x^{2}\right) ,  \label{CKH2}
\end{equation}%
\begin{equation}
E=\frac{\omega _{0}}{2}\left( e^{-2\lambda t}\ p^{2}+e^{2\lambda t}\
x^{2}\right) +\frac{\lambda }{2}\left( px+xp\right) ,\quad \frac{d}{dt}%
\left\langle E\right\rangle =0.  \label{CKH3}
\end{equation}%
By (\ref{diffops})%
\begin{equation}
\frac{d}{dt}\left\langle H\right\rangle =\left\langle \frac{\partial H}{%
\partial t}\right\rangle =\left\langle L\right\rangle .  \label{CKH4}
\end{equation}%
Applying formula (\ref{diffexp}) one gets%
\begin{eqnarray}
\frac{d}{dt}\left\langle L\right\rangle &=&2\lambda ^{2}\omega _{0}\left(
e^{-2\lambda t}\ \left\langle p^{2}\right\rangle +e^{2\lambda t}\
\left\langle x^{2}\right\rangle \right)  \label{CKH5} \\
&&+2\lambda \omega _{0}^{2}\left\langle px+xp\right\rangle  \notag
\end{eqnarray}%
and%
\begin{equation}
\frac{d}{dt}\left\langle L\right\rangle +4\omega ^{2}\left\langle
H\right\rangle =4\omega _{0}^{2}\left\langle E\right\rangle _{0}
\label{CKH6}
\end{equation}%
with the help of (\ref{CKH1}) and (\ref{CKH3}).\medskip

In view of (\ref{CKH4}) and (\ref{CKH6}) the dynamics of the Hamiltonian
expectation value $\left\langle H\right\rangle $ is governed by the
following second-order differential equation%
\begin{equation}
\frac{d^{2}}{dt^{2}}\left\langle H\right\rangle +4\omega ^{2}\left\langle
H\right\rangle =4\omega _{0}^{2}\left\langle E\right\rangle _{0}
\label{CKdiffeq}
\end{equation}%
with the unique solution given by%
\begin{equation}
\left\langle H\right\rangle =\frac{\omega ^{2}\left\langle H\right\rangle
_{0}-\omega _{0}^{2}\left\langle E\right\rangle _{0}}{\omega ^{2}}\cos
\left( 2\omega t\right) +\frac{1}{2\omega }\left\langle \frac{\partial H}{%
\partial t}\right\rangle _{0}\sin \left( 2\omega t\right) +\frac{\omega
_{0}^{2}}{\omega ^{2}}\left\langle E\right\rangle _{0}.  \label{CKsol}
\end{equation}%
The hypotheses of Lemma~4 are satisfied. Our solution allows to determine a
complete time-evolution of the expectation values of the operators $p^{2},$ $%
x^{2}$ and $px+xp.$ Further details are left to the reader.

\subsection{The Modified Caldirola-Kanai Hamiltonian}

The required operators are%
\begin{equation}
H=\frac{\omega _{0}}{2}\left( e^{-2\lambda t}\ p^{2}+e^{2\lambda t}\
x^{2}\right) -\lambda \left( px+xp\right) ,  \label{MCKH1}
\end{equation}%
\begin{equation}
L=\frac{\partial H}{\partial t}=\lambda \omega _{0}\left( -e^{-2\lambda t}\
p^{2}+e^{2\lambda t}\ x^{2}\right) =\frac{\partial H_{0}}{\partial t},
\label{MCKH2}
\end{equation}%
\begin{equation}
E=\frac{\omega _{0}}{2}\left( e^{-2\lambda t}\ p^{2}+e^{2\lambda t}\
x^{2}\right) -\frac{\lambda }{2}\left( px+xp\right) .  \label{MCKH3}
\end{equation}%
We consider the expectation value $\left\langle H_{0}\right\rangle $ of the
positive operator%
\begin{equation}
H_{0}=\frac{\omega _{0}}{2}\left( e^{-2\lambda t}\ p^{2}+e^{2\lambda t}\
x^{2}\right) .  \label{MCKH4}
\end{equation}%
In this case $H=2E-H_{0},$ and%
\begin{eqnarray}
\frac{d}{dt}\left\langle H\right\rangle &=&\left\langle \frac{\partial H}{%
\partial t}\right\rangle =\left\langle L\right\rangle =-\frac{d}{dt}%
\left\langle H_{0}\right\rangle ,  \label{MCKH5} \\
\frac{d}{dt}\left\langle L\right\rangle &=&4\omega ^{2}\left\langle
H_{0}\right\rangle -4\omega _{0}^{2}\left\langle E\right\rangle _{0},
\label{MCKH6}
\end{eqnarray}%
which results in the differential equation (\ref{CKdiffeq}) with the
explicit solution%
\begin{equation}
\left\langle H_{0}\right\rangle =\frac{\omega ^{2}\left\langle
H_{0}\right\rangle _{0}-\omega _{0}^{2}\left\langle E\right\rangle _{0}}{%
\omega ^{2}}\cos \left( 2\omega t\right) -\frac{1}{2\omega }\left\langle 
\frac{\partial H_{0}}{\partial t}\right\rangle _{0}\sin \left( 2\omega
t\right) +\frac{\omega _{0}^{2}}{\omega ^{2}}\left\langle E\right\rangle _{0}
\end{equation}%
of the initial value problem. The hypotheses of the lemma are satisfied.

\subsection{The United Model}

The related operators can be conveniently extended as follows%
\begin{equation}
H_{0}=\frac{\omega _{0}}{2}e^{\mu t}\left( e^{-2\lambda t}\
p^{2}+e^{2\lambda t}\ x^{2}\right) ,  \label{UMHamZ}
\end{equation}%
\begin{equation}
L=e^{\mu t}\left( -e^{-2\lambda t}\ p^{2}+e^{2\lambda t}\ x^{2}\right) ,
\label{UMHamL}
\end{equation}%
\begin{equation}
M=e^{\mu t}\left( px+xp\right)  \label{UMHamM}
\end{equation}%
and%
\begin{eqnarray}
E &=&H_{0}\left( t\right) +\frac{1}{2}\left( \lambda -\mu \right) M\left(
t\right)  \label{UMEnergy} \\
&=&\frac{\omega _{0}}{2}e^{\mu t}\left( e^{-2\lambda t}\ p^{2}+e^{2\lambda
t}\ x^{2}\right) +\frac{1}{2}\left( \lambda -\mu \right) e^{\mu t}\left(
px+xp\right) .  \notag
\end{eqnarray}%
Then by Lemma~2%
\begin{equation}
\frac{d}{dt}\left\langle M\right\rangle =-2\omega _{0}\left\langle
L\right\rangle ,  \label{DiffM}
\end{equation}%
\begin{equation}
\frac{d}{dt}\left\langle H_{0}\right\rangle =\omega _{0}\left( \lambda -\mu
\right) \left\langle L\right\rangle ,  \label{DiffHZ}
\end{equation}%
\begin{equation}
\frac{d}{dt}\left\langle E\right\rangle =0  \label{DiffE}
\end{equation}%
and%
\begin{equation}
\frac{d}{dt}\left\langle L\right\rangle =4\frac{\lambda -\mu }{\omega _{0}}%
\left\langle H_{0}\right\rangle +2\omega _{0}\left\langle M\right\rangle .
\label{DiffL}
\end{equation}%
In terms of the energy operator%
\begin{equation}
\frac{d}{dt}\left\langle L\right\rangle +\frac{4\omega ^{2}}{\left( \lambda
-\mu \right) \omega _{0}}\left\langle H_{0}\right\rangle =\frac{4\omega _{0}%
}{\lambda -\mu }\left\langle E\right\rangle  \label{DiffLE}
\end{equation}%
and as a result%
\begin{equation}
\frac{d^{2}}{dt^{2}}\left\langle H_{0}\right\rangle +4\omega
^{2}\left\langle H_{0}\right\rangle =4\omega _{0}^{2}\left\langle
E\right\rangle _{0},\quad \omega =\sqrt{\omega _{0}^{2}-\left( \lambda -\mu
\right) ^{2}}>0  \label{DiffEQUM}
\end{equation}%
with the unique solution of the initial value problem given by%
\begin{eqnarray}
\left\langle H_{0}\right\rangle &=&\frac{\omega ^{2}\left\langle
H_{0}\right\rangle _{0}-\omega _{0}^{2}\left\langle E\right\rangle _{0}}{%
\omega ^{2}}\cos \left( 2\omega t\right)  \label{UMSol} \\
&&+\frac{1}{2}\left( \lambda -\mu \right) \frac{\omega _{0}}{\omega }%
\left\langle L\right\rangle _{0}\sin \left( 2\omega t\right) +\frac{\omega
_{0}^{2}}{\omega ^{2}}\left\langle E\right\rangle _{0}.  \notag
\end{eqnarray}%
The hypotheses of Lemma~4 are satisfied.

\subsection{The Modified Oscillator}

The required operators are%
\begin{eqnarray}
H &=&\left( \cos t\ p+\sin t\ x\right) ^{2}  \label{MC-SS1} \\
&=&\cos ^{2}t\ p^{2}+\sin ^{2}t\ x^{2}+\sin t\cos t\ \left( px+xp\right) 
\notag \\
&=&\frac{1}{2}\left( p^{2}+x^{2}\right) +\frac{1}{2}\cos 2t\ \left(
p^{2}-x^{2}\right) +\frac{1}{2}\sin 2t\ \left( px+px\right)  \notag \\
&=&H_{0}+E\left( t\right) ,  \notag
\end{eqnarray}%
where%
\begin{equation}
H_{0}=\frac{1}{2}\left( p^{2}+x^{2}\right) ,  \label{MC-SS2}
\end{equation}%
\begin{equation}
E=E\left( t\right) =\frac{1}{2}\cos 2t\ \left( p^{2}-x^{2}\right) +\frac{1}{2%
}\sin 2t\ \left( px+px\right)  \label{MC-SS3}
\end{equation}%
and%
\begin{equation}
L=\frac{\partial H}{\partial t}=\frac{\partial E}{\partial t}=-\sin 2t\
\left( p^{2}-x^{2}\right) +\cos 2t\ \left( px+px\right) .  \label{MC-SS4}
\end{equation}%
Here%
\begin{equation}
\frac{d}{dt}\left\langle H_{0}\right\rangle =\frac{d}{dt}\left\langle
H\right\rangle =\left\langle \frac{\partial H}{\partial t}\right\rangle
=\left\langle \frac{\partial E}{\partial t}\right\rangle =\left\langle
L\right\rangle  \label{MC-SS5}
\end{equation}%
and%
\begin{equation}
\frac{d}{dt}\left\langle L\right\rangle =4\left\langle H_{0}\right\rangle .
\label{MC-SS6}
\end{equation}%
The expectation value $\left\langle H_{0}\right\rangle $ satisfies the
following differential equation%
\begin{equation}
\frac{d^{2}}{dt^{2}}\left\langle H_{0}\right\rangle =4\left\langle
H_{0}\right\rangle  \label{MC-SSeq}
\end{equation}%
with the explicit solution%
\begin{equation}
\left\langle H_{0}\right\rangle =\left\langle H_{0}\right\rangle _{0}\cosh
\left( 2t\right) +\frac{1}{2}\left\langle L\right\rangle _{0}\sinh \left(
2t\right) .  \label{MC-SSsol}
\end{equation}%
The hypotheses of Lemma~4 are satisfied.

\subsection{The Modified Damped Oscillator}

Let $\hslash =m\omega _{0}=1$ in the Hamiltonian (\ref{CJHam}): 
\begin{equation}
H=\frac{\omega _{0}}{2}\left( \frac{p^{2}}{\cosh ^{2}\left( \lambda t\right) 
}+\cosh ^{2}\left( \lambda t\right) \ x^{2}\right)  \label{CJham}
\end{equation}%
without loss of generality. The corresponding energy operator can be found
as follows%
\begin{eqnarray}
&&E=\frac{\omega _{0}}{2\cosh ^{2}\left( \lambda t\right) }p^{2}+\frac{%
\omega _{0}^{2}\sinh ^{2}\left( \lambda t\right) +\omega ^{2}}{2\omega _{0}}%
x^{2}  \label{CJEnergy} \\
&&\qquad +\frac{\lambda }{2}\tanh \left( \lambda t\right) \left(
px+xp\right) ,\qquad \frac{d}{dt}\left\langle E\right\rangle =0,  \notag
\end{eqnarray}%
in view of (\ref{energysysA})--(\ref{energysysC}) (one should replace $%
A\leftrightarrow B,$ $C\rightarrow -C$ in the momentum
representation).\medskip

Introducing the following complementary operators%
\begin{eqnarray}
H_{0} &=&\frac{p^{2}}{\cosh ^{2}\left( \lambda t\right) }+\cosh ^{2}\left(
\lambda t\right) \ x^{2},  \label{CJHZ} \\
L &=&\frac{p^{2}}{\cosh ^{2}\left( \lambda t\right) }-\cosh ^{2}\left(
\lambda t\right) \ x^{2},  \label{CJL} \\
M &=&px+xp,  \label{CJM}
\end{eqnarray}%
we get%
\begin{eqnarray}
\frac{d}{dt}\left\langle H_{0}\right\rangle &=&-2\lambda \tanh \left(
\lambda t\right) \left\langle L\right\rangle ,  \label{CJHsys} \\
\frac{d}{dt}\left\langle L\right\rangle &=&-2\lambda \tanh \left( \lambda
t\right) \left\langle H_{0}\right\rangle -2\omega _{0}\left\langle
M\right\rangle ,  \label{CJLsys} \\
\frac{d}{dt}\left\langle M\right\rangle &=&2\omega _{0}\left\langle
L\right\rangle .  \label{CJMsys}
\end{eqnarray}%
Then%
\begin{eqnarray}
E &=&\frac{\omega _{0}}{2}\left( 1-\frac{\lambda ^{2}}{2\omega _{0}^{2}\cosh
^{2}\left( \lambda t\right) }\right) H_{0}+\frac{\lambda ^{2}}{4\omega
_{0}\cosh ^{2}\left( \lambda t\right) }\ L  \label{CJenergy} \\
&&+\frac{\lambda }{2}\tanh \left( \lambda t\right) M  \notag
\end{eqnarray}%
and, eliminating $\left\langle M\right\rangle $ and $\left\langle
L\right\rangle $ from the system, one gets:%
\begin{equation}
\frac{d^{2}}{dt^{2}}\left\langle H_{0}\right\rangle -\frac{4\lambda }{\sinh
\left( 2\lambda t\right) }\frac{d}{dt}\left\langle H_{0}\right\rangle
+2\left( 2\omega ^{2}+\frac{\lambda ^{2}}{\cosh ^{2}\left( \lambda t\right) }%
\right) \left\langle H_{0}\right\rangle =8\omega _{0}\left\langle
E\right\rangle _{0}.  \label{CJEquation}
\end{equation}%
The required initial conditions:%
\begin{equation}
\left( \frac{d}{dt}\left\langle H_{0}\right\rangle \right) _{0}=0,\qquad
\left( \coth \left( \lambda t\right) \frac{d}{dt}\left\langle
H_{0}\right\rangle \right) _{0}=-2\lambda \left\langle L\right\rangle _{0}
\label{CJConditions}
\end{equation}%
follow from (\ref{CJHsys}). The unique explicit solution is given by%
\begin{eqnarray}
\left\langle H_{0}\right\rangle &=&-\lambda \frac{\lambda ^{2}\left\langle
E\right\rangle _{0}+\omega _{0}\omega ^{2}\left\langle L\right\rangle _{0}}{%
\omega _{0}\omega ^{2}\left( 2\omega ^{2}+\lambda ^{2}\right) }
\label{CJSolution} \\
&&\times \left( 2\omega \tanh \left( \lambda t\right) \sin \left( 2\omega
t\right) +\lambda \left( 1+\tanh ^{2}\left( \lambda t\right) \right) \cos
\left( 2\omega t\right) \right)  \notag \\
&&+2\left\langle E\right\rangle _{0}\frac{\omega _{0}}{\omega ^{2}}\left( 1-%
\frac{\lambda ^{2}}{2\omega _{0}^{2}\cosh ^{2}\left( \lambda t\right) }%
\right)  \notag
\end{eqnarray}%
(see appendix~D). The hypotheses of Lemma~4 are satisfied.

\subsection{The Modified Parametric Oscillator}

In the case (\ref{MPOHam}), the energy operator (\ref{EQ4}) is a positive
operator:%
\begin{equation}
\left\langle E\right\rangle =\tanh ^{2}\left( \lambda t+\delta \right) \
\left\langle p^{2}\right\rangle +\coth ^{2}\left( \lambda t+\delta \right) \
\left\langle x^{2}\right\rangle =\left\langle E\right\rangle _{0}>0.
\label{MPOEnergy}
\end{equation}%
The related operators are%
\begin{eqnarray}
L &=&\tanh ^{2}\left( \lambda t+\delta \right) \ p^{2}-\coth ^{2}\left(
\lambda t+\delta \right) \ x^{2},  \label{MPOL} \\
M &=&px+xp,  \label{MPOM} \\
H &=&\frac{\omega }{2}\ E+\frac{\lambda }{\sinh \left( 2\lambda t+2\delta
\right) }\ M  \label{MPOH}
\end{eqnarray}%
with%
\begin{equation}
\frac{d}{dt}\left\langle L\right\rangle =-2\omega \left\langle
M\right\rangle ,\qquad \frac{d}{dt}\left\langle M\right\rangle =-2\omega
\left\langle L\right\rangle .  \label{MPODiffLM}
\end{equation}%
From here%
\begin{equation}
\frac{d^{2}}{dt^{2}}\left\langle L\right\rangle +4\omega ^{2}\left\langle
L\right\rangle =0,\qquad \frac{d^{2}}{dt^{2}}\left\langle M\right\rangle
+4\omega ^{2}\left\langle M\right\rangle =0,  \label{MPODiffEqLM}
\end{equation}%
which determines the time-evolution of the expectation values.

\subsection{Parametric Oscillators}

In general the Lewis--Riesenfeld quadratic invariant (\ref{ParamQuadInv})
for the parametric oscillator (\ref{ParametHam}) is obviously a positive
operator for real-valued solutions of the Ermakov equation (\ref%
{NonlinearODE}) that satisfies the conditions of our lemma.

\subsection{General Quadratic Hamiltonian}

In the case of Hamiltonian (\ref{GenHam}) applying formula (\ref{GenSys}) to
the operators, $O=\left\{ p^{2},x^{2},px+xp\right\} ,$ one obtains \cite%
{Cor-Sot:Sua:Sus}:%
\begin{equation}
\frac{d}{dt}\left( 
\begin{array}{c}
\left\langle p^{2}\right\rangle \smallskip \\ 
\left\langle x^{2}\right\rangle \smallskip \\ 
\left\langle px+xp\right\rangle%
\end{array}%
\right) =\left( 
\begin{array}{ccc}
-3c\left( t\right) -d\left( t\right) & 0 & -2b\left( t\right) \smallskip \\ 
0 & c\left( t\right) +3d\left( t\right) & 2a\left( t\right) \smallskip \\ 
4a\left( t\right) & -4b\left( t\right) & -c\left( t\right) +d\left( t\right)%
\end{array}%
\right) \left( 
\begin{array}{c}
\left\langle p^{2}\right\rangle \smallskip \\ 
\left\langle x^{2}\right\rangle \smallskip \\ 
\left\langle px+xp\right\rangle%
\end{array}%
\right) .  \label{GSystem}
\end{equation}%
This system has a unique solution for suitable coefficients \cite{HilleODE},
which allows one to apply Lemma~4, say, for the positive operator $x^{2}.$
Our Theorem~1 provides another choice of positive operators. On the second
thought a positive integral (\ref{b3}) determines time-evolution of the
squared norm and guarantees uniqueness in $L^{2}\left( 
\mathbb{R}
\right) .$ Details are left to the reader.\medskip

\noindent \textbf{Acknowledgments.\/} We thank Professor Carlos Castillo-Ch%
\'{a}vez, Professor Victor V. Dodonov, Professor Vladimir~I. Man'ko and
Professor Kurt Bernardo Wolf for support, valuable discussions and
encouragement. The authors are indebted to Professor George A.~Hagedorn for
kindly pointing out the papers \cite{Hag:Loss:Slaw} and \cite{Haged98} to
our attention. We thank David Murillo for help. The authors are grateful to
Professor Peter~G.~L.~Leach for careful reading of the manuscript --- his
numerous suggestions have helped to improve the presentation. One of the
authors (RCS) is supported by the following National Science Foundation
programs: Louis Stokes Alliances for Minority Participation (LSAMP): NSF
Cooperative Agreement No. HRD-0602425 (WAESO LSAMP Phase IV); Alliances for
Graduate Education and the Professoriate (AGEP): NSF Cooperative Agreement
No. HRD-0450137 (MGE@MSA AGEP Phase II).

\appendix

\section{The Ehrenfest Theorems}

Application of formula (\ref{diffops}) to the position $x$ and momentum $p$
operators allows one to derive the Ehrenfest theorem \cite{Ehrenfest}, \cite%
{Merz}, \cite{Schiff} for the models of oscillators under
consideration.\medskip

For the Caldirola-Kanai Hamiltonian (\ref{CKham}) one gets%
\begin{equation}
\frac{d}{dt}\left\langle x\right\rangle =\omega _{0}e^{-2\lambda
t}\left\langle p\right\rangle ,\qquad \frac{d}{dt}\left\langle
p\right\rangle =-\omega _{0}e^{2\lambda t}\left\langle x\right\rangle .
\label{a1}
\end{equation}%
Elimination of the expectation value $\left\langle p\right\rangle $ from
this system results in the classical equation of motion for a damped
oscillator \cite{BatemanPDE}, \cite{Lan:Lif}:%
\begin{equation}
\frac{d^{2}}{dt^{2}}\left\langle x\right\rangle +2\lambda \frac{d}{dt}%
\left\langle x\right\rangle +\omega _{0}^{2}\left\langle x\right\rangle =0.
\label{a2}
\end{equation}%
For the modified Caldirola-Kanai Hamiltonian (\ref{modCKham}) the system%
\begin{equation}
\frac{d}{dt}\left\langle x\right\rangle =\omega _{0}e^{-2\lambda
t}\left\langle p\right\rangle -2\lambda \left\langle x\right\rangle ,\qquad 
\frac{d}{dt}\left\langle p\right\rangle =-\omega _{0}e^{2\lambda
t}\left\langle x\right\rangle +2\lambda \left\langle p\right\rangle
\label{a3}
\end{equation}%
gives the same classical equation.\medskip

In the case of the united model (\ref{UMHam}) one should use the
differentiation formula (\ref{DiffOper}). Then%
\begin{equation}
\frac{d}{dt}\left\langle x\right\rangle =\omega _{0}e^{-2\lambda
t}\left\langle p\right\rangle -2\mu \left\langle x\right\rangle ,\qquad 
\frac{d}{dt}\left\langle p\right\rangle =-\omega _{0}e^{2\lambda
t}\left\langle x\right\rangle  \label{a3a}
\end{equation}%
and the second order equations are given by%
\begin{equation}
\frac{d^{2}}{dt^{2}}\left\langle x\right\rangle +\ 2\left( \lambda +\mu
\right) \frac{d}{dt}\left\langle x\right\rangle +\left( \omega
_{0}^{2}+4\lambda \mu \right) \left\langle x\right\rangle =0,  \label{a3x}
\end{equation}%
\begin{equation}
\frac{d^{2}}{dt^{2}}\left\langle p\right\rangle +\ 2\left( \mu -\lambda
\right) \frac{d}{dt}\left\langle p\right\rangle +\omega _{0}^{2}\left\langle
p\right\rangle =0.  \label{a3p}
\end{equation}%
The general solutions are%
\begin{eqnarray}
\left\langle x\right\rangle &=&Ae^{-\left( \lambda +\mu \right) t}\sin
\left( \omega t+\delta \right) ,  \label{a3xsol} \\
\left\langle p\right\rangle &=&Be^{\left( \lambda -\mu \right) t}\sin \left(
\omega t+\gamma \right) ,  \label{a3psol}
\end{eqnarray}%
where $\omega =\sqrt{\omega _{0}^{2}-\left( \lambda -\mu \right) ^{2}}>0.$%
\medskip\ 

In a similar fashion for a modified oscillator with the Hamiltonian (\ref%
{mod1}) we obtain%
\begin{eqnarray}
\frac{d}{dt}\left\langle x\right\rangle &=&2\cos ^{2}t\ \left\langle
p\right\rangle +2\sin t\cos t\ \left\langle x\right\rangle ,  \label{a4} \\
\frac{d}{dt}\left\langle p\right\rangle &=&-2\sin ^{2}t\ \left\langle
x\right\rangle -2\sin t\cos t\ \left\langle p\right\rangle .  \label{a5}
\end{eqnarray}%
Then%
\begin{equation}
\frac{d^{2}}{dt^{2}}\left\langle x\right\rangle +\ 2\tan t\frac{d}{dt}%
\left\langle x\right\rangle -2\left\langle x\right\rangle =0,  \label{a6}
\end{equation}%
which coincides with the characteristic equation (\ref{in6}) in this case 
\cite{Cor-Sot:Sus}.\medskip

In the case of the damped oscillator of Chru\'{s}ci\'{n}ski and Jurkowski
one obtains%
\begin{eqnarray}
\frac{d}{dt}\left\langle x\right\rangle &=&\frac{\left\langle p\right\rangle 
}{m\cosh ^{2}\left( \lambda t\right) },  \label{CJx} \\
\frac{d}{dt}\left\langle p\right\rangle &=&-m\omega _{0}^{2}\cosh ^{2}\left(
\lambda t\right) \left\langle x\right\rangle .  \label{CJp}
\end{eqnarray}%
The Ehrenfest theorems coincide with the Newtonian equations of motion \cite%
{Chru:Jurk}:%
\begin{eqnarray}
\frac{d^{2}}{dt^{2}}\left\langle x\right\rangle +\ 2\lambda \tanh \left(
\lambda t\right) \frac{d}{dt}\left\langle x\right\rangle +\omega
_{0}^{2}\left\langle x\right\rangle &=&0,  \label{CJEhr} \\
\frac{d^{2}}{dt^{2}}\left\langle p\right\rangle -\ 2\lambda \tanh \left(
\lambda t\right) \frac{d}{dt}\left\langle p\right\rangle +\omega
_{0}^{2}\left\langle p\right\rangle &=&0  \label{CJEhrP}
\end{eqnarray}%
with the general solutions given by%
\begin{equation}
\left\langle x\right\rangle =A\frac{\sin \left( \omega t+\delta \right) }{%
\cosh \left( \lambda t\right) },\qquad \omega =\sqrt{\omega _{0}^{2}-\lambda
^{2}}>0,  \label{CJEhrsol}
\end{equation}%
\begin{equation}
\left\langle p\right\rangle =B\left( \lambda \cos \left( \omega t+\delta
\right) \sinh \left( \lambda t\right) +\omega \sin \left( \omega t+\delta
\right) \cosh \left( \lambda t\right) \right) ,  \label{CJEhrsolP}
\end{equation}%
respectively. It is worth noting that both equations (\ref{a2}) and (\ref%
{CJEhr}) give the same frequency of oscillations for the damped motion; see 
\cite{Chru:Jurk} for more details.\medskip

Combining all models together for the general quadratic Hamiltonian (\ref%
{GenHam}):%
\begin{equation}
\frac{d}{dt}\left\langle x\right\rangle =2a\left( t\right) \ \left\langle
p\right\rangle +2d\left( t\right) \ \left\langle x\right\rangle ,\quad \frac{%
d}{dt}\left\langle p\right\rangle =-2b\left( t\right) \ \left\langle
x\right\rangle -2c\left( t\right) \ \left\langle p\right\rangle
\label{GenXP}
\end{equation}%
with the help of (\ref{DiffOper}). The Newtonian-type equation of motion for
the expectation values has the form%
\begin{equation}
\frac{d^{2}}{dt^{2}}\left\langle x\right\rangle -\tau \left( t\right) \frac{d%
}{dt}\left\langle x\right\rangle +4\sigma \left( t\right) \left\langle
x\right\rangle =0  \label{GenEhrenfest}
\end{equation}%
with%
\begin{equation}
\tau \left( t\right) =\frac{a^{\prime }}{a}-2c+2d,\qquad \sigma \left(
t\right) =ab-cd+\frac{d}{2}\left( \frac{a^{\prime }}{a}-\frac{d^{\prime }}{d}%
\right) .  \label{TauSigmaEhr}
\end{equation}%
In order to explain a connection with the characteristic equation (\ref{in6}%
)--(\ref{in7}) we temporarily replace $c\rightarrow c_{0}$ and $d\rightarrow
d_{0}$ in the original Hamiltonian (\ref{in1}). Then it takes the standard
form (\ref{GenHam}), if $c_{0}=c+d$ and $d_{0}=c.$ Using the new notations
in (\ref{in6})--(\ref{in7}) we find%
\begin{equation}
\tau -\tau _{0}=4\left( d-c\right) ,\quad \sigma -\sigma _{0}=\frac{a}{2}%
\left( \frac{c-d}{a}\right) ^{\prime }.  \label{DiffSigmaTau}
\end{equation}%
Therefore our characteristic equation (\ref{in6}) coincides with the
corresponding Ehrenfest theorem (\ref{GenEhrenfest}) only in the case of
self-adjoint Hamiltonians, when $c=d$ (or $c_{0}=2d_{0}$). The united model
shows that these equations are different otherwise.

\section{The Heisenberg Uncertainty Relation Revisited}

A detailed review with an extensive list of references is given in Refs.~%
\cite{Dodonov:Man'koFIAN87CohSt} and \cite{JvNeumann} (see also \cite%
{Mand:Karpov:Cerf}). We only discuss the Heisenberg uncertainty relation for
the position $x$ and momentum $p=-i\partial /\partial x$ operators (in the
units of $\hslash $) in the case of the general quadratic Hamiltonian (\ref%
{GenHam}). By our Lemma~2 the simplest integral of motion is given by%
\begin{equation}
E_{0}=\exp \left( \int_{0}^{t}\left( c\left( \tau \right) -d\left( \tau
\right) \right) \ d\tau \right) \ \left( px-xp\right)  \label{b1}
\end{equation}%
with%
\begin{equation}
\left[ x,p\right] =xp-px=i.  \label{b2}
\end{equation}%
This implies the following time evolution: 
\begin{equation}
\left\langle \psi ,\psi \right\rangle =\exp \left( \int_{0}^{t}\left(
d\left( \tau \right) -c\left( \tau \right) \right) \ d\tau \right)
\left\langle \psi ,\psi \right\rangle _{0}  \label{b3}
\end{equation}%
of the squared norm of the wave functions.\medskip

With the expectation values%
\begin{equation}
\overline{x}=\frac{\left\langle x\right\rangle }{\left\langle 1\right\rangle 
}=\frac{\left\langle \psi ,x\psi \right\rangle }{\left\langle \psi ,\psi
\right\rangle },\qquad \overline{p}=\frac{\left\langle p\right\rangle }{%
\left\langle 1\right\rangle }=\frac{\left\langle \psi ,p\psi \right\rangle }{%
\left\langle \psi ,\psi \right\rangle }  \label{b4}
\end{equation}%
and the operators%
\begin{equation}
\Delta x=x-\overline{x},\qquad \Delta p=p-\overline{p}  \label{b5}
\end{equation}%
let us consider%
\begin{eqnarray}
0 &\leq &\left\langle \left( \Delta x+i\lambda \Delta p\right) \psi ,\left(
\Delta x+i\lambda \Delta p\right) \psi \right\rangle   \label{b6} \\
&=&\left\langle \psi ,\left( \Delta x-i\lambda \Delta p\right) \left( \Delta
x+i\lambda \Delta p\right) \psi \right\rangle   \notag \\
&=&\left\langle \left( \Delta x\right) ^{2}\right\rangle -\lambda
\left\langle 1\right\rangle +\lambda ^{2}\left\langle \left( \Delta p\right)
^{2}\right\rangle   \notag
\end{eqnarray}%
for a real parameter $\lambda .$ Here we have used the operator identity%
\begin{equation}
\left( \Delta x-i\lambda \Delta p\right) \left( \Delta x+i\lambda \Delta
p\right) =\left( \Delta x\right) ^{2}-\lambda +\lambda ^{2}\left( \Delta
p\right) ^{2}.  \label{b7}
\end{equation}%
Then one gets%
\begin{equation}
\left\langle \left( \Delta p\right) ^{2}\right\rangle \left\langle \left(
\Delta x\right) ^{2}\right\rangle \geq \frac{1}{4}\left\langle
1\right\rangle ^{2}=\frac{1}{4}\exp \left( 2\dint_{0}^{t}\left( d\left( \tau
\right) -c\left( \tau \right) \right) \ d\tau \right) ,  \label{b8}
\end{equation}%
if $\left\langle 1\right\rangle _{0}=\left\langle \psi ,\psi \right\rangle
_{0}=1.$ For the standard deviations:%
\begin{equation}
\left( \delta p\right) ^{2}=\frac{\left\langle \left( \Delta p\right)
^{2}\right\rangle }{\left\langle 1\right\rangle }=\overline{\left(
p^{2}\right) }-\left( \overline{p}\right) ^{2},  \label{bdp}
\end{equation}%
\begin{equation}
\left( \delta x\right) ^{2}=\frac{\left\langle \left( \Delta x\right)
^{2}\right\rangle }{\left\langle 1\right\rangle }=\overline{\left(
x^{2}\right) }-\left( \overline{x}\right) ^{2},  \label{bdx}
\end{equation}%
we finally obtain%
\begin{equation}
\delta p\ \delta x\geq \frac{1}{2}  \label{Heisenberg}
\end{equation}%
in the units of $\hslash .$ It is worth noting that Eq.~(\ref{Heisenberg})
is derived, in fact, for any operators $x$ and $p$ with the commutator (\ref%
{b2}) --- the structure of the quadratic Hamiltonian (\ref{GenHam}) has only
been used in the norm (\ref{b3}). Time-evolution of the standard deviations (%
\ref{bdp})--(\ref{bdx}) will be discussed elsewhere.

\section{Linear Integrals of Motion: The Dodonov--Malkin--Man'ko--Trifonov
Invariants}

All invariants of the form%
\begin{equation}
P=A\left( t\right) p+B\left( t\right) x+C\left( t\right)   \label{LinInts}
\end{equation}%
for the general quadratic Hamiltonian (\ref{GenHam}) can be found as follows
(see, for example, \cite{Dod:Man79}, \cite{Dodonov:Man'koFIAN87}, \cite%
{Malkin:Man'ko79}, \cite{Malk:Man:Trif73} and references therein). Use of
the differentiation formula (\ref{DiffOper}) results in the following system:%
\begin{eqnarray}
\frac{dA}{dt} &=&2c\left( t\right) A-2a\left( t\right) B,  \label{LinCA} \\
\frac{dB}{dt} &=&2b\left( t\right) A-2d\left( t\right) B,  \label{LinCB} \\
\frac{dC}{dt} &=&\left( c\left( t\right) -d\left( t\right) \right) C.
\label{LinCC}
\end{eqnarray}%
The last equation is explicitly integrated and elimination of $B$ and $A$
from (\ref{LinCA}) and (\ref{LinCB}), respectively, gives the second-order
equations:%
\begin{eqnarray}
A^{\prime \prime }-\left( \frac{a^{\prime }}{a}+2c-2d\right) A^{\prime
}+4\left( ab-cd+\frac{c}{2}\left( \frac{a^{\prime }}{a}-\frac{c^{\prime }}{c}%
\right) \right) A &=&0,  \label{EqC} \\
B^{\prime \prime }-\left( \frac{b^{\prime }}{b}+2c-2d\right) B^{\prime
}+4\left( ab-cd-\frac{d}{2}\left( \frac{b^{\prime }}{b}-\frac{d^{\prime }}{d}%
\right) \right) B &=&0.  \label{EqBC}
\end{eqnarray}%
The first is equivalent to our characteristic equation (\ref{in6})--(\ref%
{in7}) and coincides with the Ehrenfest theorem (\ref{GenEhrenfest})--(\ref%
{TauSigmaEhr}) when $c\leftrightarrow d.$\medskip 

Thus the linear quantum invariants are given by%
\begin{equation}
P=A\left( t\right) p+\frac{2c\left( t\right) A\left( t\right) -A^{\prime
}\left( t\right) }{2a\left( t\right) }x+C_{0}\exp \left( \int_{0}^{t}\left(
c\left( \tau \right) -d\left( \tau \right) \right) \ d\tau \right) ,
\label{GenLinInv}
\end{equation}%
where $A\left( t\right) $ is a general solution of equation (\ref{EqC})
depending upon two parameters and $C_{0}$ is the third constant. Our
Theorem~1 gives a similar description of the quadratic invariants in terms
of solutions of the auxiliary equation (\ref{AuxEq}). Relations between
linear and quadratic invariants are analyzed in \cite{Suslov10}.
Group-theoretical applications are discussed in \cite{Cher:Man08}, \cite%
{Dod:Mal:Man75}, \cite{Dodonov:Man'koFIAN87}, \cite{Malkin:Man'ko79}, \cite%
{Dod:Man79}, \cite{Lewis:Riesen69}, \cite{Malk:Man:Trif73} and elsewhere.

\section{An Elementary Differential Equation}

The nonhomogeneous differential equation of the form%
\begin{equation}
y^{\prime \prime }-\frac{4\lambda }{\sinh \left( 2\lambda t+2\gamma \right) }%
y^{\prime }+\left( \omega ^{2}+\frac{2\lambda ^{2}}{\cosh ^{2}\left( \lambda
t+\gamma \right) }\right) y=1  \label{DEquation1}
\end{equation}%
($\omega ,$ $\lambda $ and $\gamma $ are some parameters) has the following
general solution:%
\begin{equation}
y=C_{1}y_{1}\left( t\right) +C_{2}y_{2}\left( t\right) +Y\left( t\right) ,
\label{DESolution}
\end{equation}%
where $C_{1}$ and $C_{2}$ are constants,%
\begin{eqnarray}
y_{1} &=&\omega \tanh \left( \lambda t+\gamma \right) \cos \left( \omega
t\right) -\lambda \left( 1+\tanh ^{2}\left( \lambda t+\gamma \right) \right)
\sin \left( \omega t\right) ,  \label{DESoly1} \\
y_{2} &=&\omega \tanh \left( \lambda t+\gamma \right) \sin \left( \omega
t\right) +\lambda \left( 1+\tanh ^{2}\left( \lambda t+\gamma \right) \right)
\cos \left( \omega t\right)  \label{DESoly2}
\end{eqnarray}%
are the fundamental solutions of the corresponding homogeneous equation with
the Wronskian given by%
\begin{equation}
W\left( y_{1},y_{2}\right) =\omega \left( \omega ^{2}+4\lambda ^{2}\right)
\tanh ^{2}\left( \lambda t+\gamma \right) ,  \label{Wronskiany1y2}
\end{equation}%
and%
\begin{equation}
Y=\frac{1}{\omega ^{2}}\left( 1-\frac{2\lambda ^{2}}{\left( \omega
^{2}+4\lambda ^{2}\right) \cosh ^{2}\left( \lambda t+\gamma \right) }\right)
\label{DESoly3}
\end{equation}%
is a particular solution of the nonhomogeneous equation.\medskip

One can also verify that functions:%
\begin{eqnarray}
z_{1} &=&\omega \cos \left( \omega t\right) -\lambda \coth \left( \lambda
t+\gamma \right) \sin \left( \omega t\right) ,  \label{DESolz1} \\
z_{2} &=&\omega \sin \left( \omega t\right) +\lambda \coth \left( \lambda
t+\gamma \right) \cos \left( \omega t\right)  \label{DESolz2}
\end{eqnarray}%
with $W\left( z_{1},z_{2}\right) =\omega \left( \omega ^{2}+\lambda
^{2}\right) $ are fundamental solutions of the following equation:

\begin{equation}
z^{\prime \prime }+\left( \omega ^{2}-\frac{2\lambda ^{2}}{\sinh ^{2}\left(
\lambda t+\gamma \right) }\right) z=0  \label{DEquation2}
\end{equation}%
and then carry out the substitution $y=z\tanh \left( \lambda t+\gamma
\right) .$ Details are left to the reader. The particular solution of the
nonhomogeneous equation can be found by the variation of parameters and/or
verified by the substitution. Review of other integrable second-order
differential equations is given in \cite{Ermakov}.

\end{document}